%% file: main_arXiv.tex
\documentclass{article}
\usepackage[utf8]{inputenc}
\usepackage{authblk}
\usepackage{amsmath,amssymb,amsthm}
\usepackage{mathtools,nccmath}
\usepackage{appendix}
\usepackage[ruled]{algorithm2e}

\usepackage{booktabs}

\newtheorem{theorem}{Theorem}
\newtheorem{lemma}{Lemma}
\newtheorem{corollary}{Corollary}

\theoremstyle{plain}

\newtheorem{definition}{Definition}

\theoremstyle{remark}
\newtheorem{remark}{Remark}

\usepackage{subcaption}
\usepackage{geometry}
\usepackage{hyperref}
\usepackage{float}
\usepackage{enumitem}
\usepackage{comment}
\usepackage{multirow}
\setlist{topsep=1pt,itemsep=0pt,labelwidth=8pt,labelsep=5pt,leftmargin=16pt}
\newcommand{\tran}{\mathsf{T}}

\DeclareMathOperator*{\argmin}{arg\,min}

\usepackage{xcolor}

\title{Zeroth-Order Feedback-Based Optimization for Distributed Demand Response}
\date{}
\author{Ruiyang Jin, Yujie Tang, Jie Song}
\affil{Peking University, Department of Industrial Engineering \& Management\\
\{\href{mailto:jry@pku.edu.cn}{jry}, \href{mailto:yujietang@pku.edu.cn}{yujietang}, \href{mailto:jie.song@pku.edu.cn}{jie.song}\}@pku.edu.cn}

\begin{document}

\maketitle
\begin{abstract}
Distributed demand response is a typical distributed optimization problem that requires coordination among multiple agents to satisfy demand response requirements. However, existing distributed algorithms for this problem still face challenges such as unknown system models, nonconvexity, privacy issues, etc. To address these challenges, we propose and analyze two distributed algorithms, in which the agents do not share their information and instead perform local updates using zeroth-order feedback information to estimate the gradient of the global objective function. One algorithm applies to problems with general convex and compact feasible sets but has higher oracle complexity bounded by $\mathcal{O}(d/\epsilon^2)$, while the other algorithm achieves lower complexity bound $\mathcal{O}(d/\epsilon)$ but is only applicable to problems with box constraints. We conduct empirical experiments to validate their performance.
\end{abstract}

\input{introduction_new}

\input{problem_formulation_new}

\input{algorithm_design_new}

\input{algorithm_analysis_new}

\input{case_study_new}

\section{Conclusion}

We studied a distributed demand response problem in which a mathematical model of the system's physics is not available. We proposed two distributed zeroth-order algorithms, 2-ZFGD and RZFCD, to address the issue of lacking the system model. Furthermore, the two algorithms do not require the agents to upload their load or preference information to the aggregator, which can help preserve the agents' privacy. We provided theoretical analysis of the two algorithms for both the convex case and the nonconvex case, and compared their advantages and disadvantages in detail. Numerical experiments were conducted to verify the performance of the proposed algorithms.

We emphasize that this work is only a starting point that illustrates the potential of applying zeroth-order optimization methods to distributed demand response, and there are still questions and issues that need to be addressed before the algorithms can be actually implemented in real systems, such as how to choose the algorithm parameters, how to ensure safety of random exploration during the optimization procedure, how to deal with measurement noise and error when accessing zeroth-order information, how to handle temporal coupling introduced by energy storage, etc. From a theoretical perspective, it will be interesting to see whether we can further reduce the complexity of the algorithms by, e.g., employing Nesterov's acceleration techniques or exploiting the structural properties of the power grid. It would be also interesting to investigate quantitative privacy guarantees for zeroth-order optimization methods.

\bibliographystyle{unsrt}
\bibliography{main_arXiv}
\appendix
\input{appendix_new}
\end{document}

%% file: introduction_new.tex
\section{Introduction}

With the higher proportion of renewable energy integrated into the smart grid \cite{phuangpornpitak2013opportunities}, the real-time generation-demand balance in power systems necessitates the management of more flexible resources on both the generation side and the demand side. Demand response (DR) is an important approach for demand-side management that coordinates the end-users' electricity usage to change from their normal patterns by incentive-based or price-based methods \cite{albadi2008summary}. The benefits of DR include reduction in system operating costs and generation capacity requirements, increased economic efficiency, etc.~\cite{pinson2014benefits}, and can be enjoyed by both the power grid and end users.

Distributed demand response (DDR) aims to coordinate different types of distributed energy resources on the demand side, such as residential, commercial, and industrial loads, distributed generators, energy storage, etc. \cite{safdarian2014distributed,roldan2019improving}. 
The distributed manner of coordination in DDR allows the distributed users to participate in the computation and decision procedure by communicating with neighboring users or aggregator iteratively. Existing works have proposed to adapt various distributed optimization methods for developing DDR algorithms, including alternative direction multiplier method (ADMM) \cite{tan2014optimal,scott2015distributed,kou2020scalable,shi2021distributed}, dual decomposition \cite{hu2018distributed}, consensus-based methods \cite{chen2014distributed, qin2018consensus}, game theoretic approaches~\cite{yaagoubi2014user, latifi2017fully}, etc.; see Section~\ref{subsec:related_works} for a review of related works.

However, there still remain challenges in the design of DDR algorithms that are yet to be fully addressed. Many formulations of DDR incorporate voltage and power flow constraints to ensure safe operation \cite{scott2015distributed,shi2021distributed}. Considering that most existing DDR algorithms are gradient-based or model-based, the optimization procedure of DDR will necessarily require knowledge of a detailed model of the power grid. However, in practical scenarios, it can be challenging to develop a detailed mathematical model describing the power grid's physics that is both accurate and computationally tractable, especially when the numbers of buses and lines in the power grid are huge and when the connected devices are highly heterogeneous~\cite{prostejovsky2016distribution,li2020learning}. Besides, the nonlinearity and nonconvexity of the power flow add further layers of difficulty in deriving performance guarantees of DDR algorithms. Furthermore, with the increasing level of digitization in the smart grid, privacy concerns become one of the main obstacles that hinder the development and adoption of DR~\cite{pop2020blockchain}. Users are concerned about privacy leakage when participating in DR and sharing load information with others \cite{tsaousoglou2020truthful}, but most distributed optimization algorithms require information sharing among agents or between agents and the aggregator, which can be a major source of privacy leakage. These challenges motivate our study of zeroth-order feedback-based optimization algorithms for distributed demand response.

\subsection{Related Works}
\label{subsec:related_works}

\paragraph*{Distributed optimization algorithms for DDR} As mentioned before, existing literature has adapted various distributed optimization techniques for designing DDR algorithms, and here we only provide an inexhaustive review. For example, the works ~\cite{tan2014optimal,scott2015distributed,kou2020scalable,shi2021distributed} employed distributed ADMM to decompose the full DDR problem into iterative subproblems that are solved successively on the users' side and the aggregator/utility's side, while \cite{hu2018distributed} used dual decomposition to formulate the subproblems. The consensus method is another class of distributed optimization techniques that drives local copies of the decision variable to simultaneously achieve consensus and optimality and was employed by works including~\cite{chen2014distributed, qin2018consensus}. Game theoretic approaches, on the other hand, model DDR as a game involving an operator and multiple distributed agents, and the goal is to achieve equilibrium by strategic interactions~\cite{yaagoubi2014user, latifi2017fully}. We mention that most of these algorithms require agents to share their load or preference information with others and require knowledge of a mathematical model of the system.

\paragraph*{Privacy preservation in DDR} The information sharing among agents and the aggregator in existing DDR algorithms can be a major source of privacy leakage, and existing literature has proposed different methods for privacy preservation. The techniques for privacy protection in DDR include information encryption~\cite{pop2020blockchain,li2013eppdr,rahman2017privacy} and differential privacy-based strategies~\cite{hassan2019differential,liu2017achieving}. Information encryption strategies prevent unauthorized users from accessing the encrypted information but have very high computational overhead and require auxiliary devices that may be costly. Differential privacy-based methods, on the other hand, ensure a controllable degree of privacy preservation with low computational overhead but data accuracy will be impaired.

\paragraph*{Zeroth-order optimization}
Zeroth-order gradient estimation is a promising technique in zeroth-order/derivative-free optimization that has recently attracted much attention for designing optimization algorithms in the model-free setting. The main idea behind zeroth-order gradient estimation is to construct a stochastic gradient from zeroth-order function values at randomly explored points, leading to zeroth-order algorithms that enjoy similar convergence guarantees as first-order methods~\cite{duchi2015optimal,nesterov2017random}. Due to its close relation to stochastic first-order methods, zeroth-order gradient estimation has also been successfully applied in distributed zeroth-order optimization~\cite{hajinezhad2019zone,wang2022distributed,tang2023zeroth}. We refer to~\cite{larson2019derivative,liu2020primer} for more detailed surveys of zeroth-order optimization methods and their applications.

\begin{table}
\centering
\begin{tabular}{c|c|c|c}
\toprule
& Feasible region & Complexity (convex) & Complexity (nonconvex) \\
\midrule
2-ZFGD & Convex \& compact & $\mathcal{O}\!\left(\mfrac{d}{\epsilon^2}\right)$ & $\mathcal{O}\!\left(\mfrac{d}{\epsilon^2}\right)$ \\
\midrule
RZFCD & Box & $\mathcal{O}\!\left(\mfrac{d}{\epsilon}\right)$ & $\mathcal{O}\!\left(\mfrac{d}{\epsilon}\right)$ \\
\bottomrule
\end{tabular}
\caption{Comparison of the two proposed DDR algorithms. The complexity is measured in terms of the number of zeroth-order queries needed to achieve $\min_{k\leq K}\mathbb{E}[F(x(k))-F^\ast]\leq\epsilon$ (convex setting) or $\min_{k< K}\mathbb{E}[\|\mathfrak{g}(x(k);M)\|^2]\leq\epsilon$ (nonconvex setting); see Section~\ref{section_3} for detailed definitions. $d$ is the dimension of the decision variable.}
\label{table:contribution}
\end{table}

\subsection{Our Contributions}
In this paper, we study distributed zeroth-order methods for distributed demand response. We formulate a generalized DDR problem, in which an aggregator needs to coordinate multiple distributed agents to minimize global and local objectives. The technical contributions of this paper can be summarized as follows:
\begin{enumerate}
\item We design two DDR algorithms that incorporate zeroth-order gradient estimation techniques to address the issue of lacking system models. The two proposed algorithms do not involve gradient computation of the global objective that requires a mathematical model of the power grid, but instead exploit observed feedback values (zeroth-order information) of the global objective to produce a stochastic gradient estimator. Furthermore, the participating agents do not need to upload their load or preference information, which helps to preserve their privacy during the optimization procedure.

\item We analyze and compare the performance of the two proposed DDR algorithms, which is summarized in Table~\ref{table:contribution}. Specifically, we derive the complexity bounds for the two algorithms for both the \emph{convex} and the \emph{nonconvex} settings, which quantitatively characterize the efficiency of the two algorithms.
The 2-ZFGD algorithm applies to situations where the feasible region is a general compact and convex set, but its complexities for both the convex and the nonconvex settings are upper bounded by $\mathcal{O}(d/\epsilon^2)$ which is inferior. Whereas the RZFCD algorithm achieves better complexity bounds $\mathcal{O}(d/\epsilon)$ for both the convex and the nonconvex settings, but only has performance guarantees when the feasible region is a multi-dimensional box.
\end{enumerate}


To the best of our knowledge, existing literature has not yet proposed distributed algorithms that can handle black-box system behavior and avoid sharing agents' preferences to solve the DDR problem. Moreover, our analysis includes both the convex and the nonconvex settings for constrained problems, and the design and analysis of RZFCD shed light on how to close the gap between zeroth-order smooth unconstrained and constrained optimization (see the discussion after Corollary~\ref{corollary:RZFCD_complexity_nonconvex}), which we believe has independent theoretical interest for researchers in the area of general zeroth-order optimization.

\paragraph*{Notations} For a subset $S\subseteq\mathbb{R}^p$ and a real number $\alpha\in\mathbb{R}$, denote $\alpha S\coloneqq \{\alpha x\mid x\in S\}$. The interior of $S\subseteq\mathbb{R}^p$ will be denoted by $\operatorname{int} S$. For a multivariate function $h(x)$ with $x=(x_1,\ldots,x_N)$ and each $x_i\in\mathbb{R}^{p_i}$, we let $\nabla_i h(x)$ denote the partial gradient of $h$ with respect to $x_i$ evaluated at $x$. To distinguish between subvectors and entries of a vector $x$, we use $i,j$ to denote indices of subvectors $x_i,x_j$, while Greek letters $\alpha,\beta$ are reserved for indices of entries $x_\alpha,x_\beta$. We let $\langle\cdot,\cdot\rangle$ denote the standard inner product and let $\|\cdot\|$ denote the $\ell_2$ norm on $\mathbb{R}^p$. 
The closed unit ball in $\mathbb{R}^p$ will be denoted by $\mathbb{B}_p\coloneqq \{x\in\mathbb{R}^p: \|x\|\leq 1\}$, and the unit sphere in $\mathbb{R}^p$ will be denoted by $\mathbb{S}_{p-1}\coloneqq \{x\in\mathbb{R}^p: \|x\|=1\}$.

%% file: problem_formulation_new.tex
\section{Problem Formulation and Preliminaries}

\subsection{Formulation of the Distributed Demand Response Problem}

Consider an aggregator trying to satisfy certain DR requirements from a higher-level grid operator by coordinating $N$ distributed agents. On the one hand, the DR program needs to meet a certain global goal such as curtailing specific amounts of load, minimizing peak-to-average rate, etc. On the other hand, the discomfort losses or costs of users induced by participating in DR should also be considered because DR is acceptable only based on low influence on user experience. Generally, we can model DDR as an optimization problem, whose objective function consists of a global cost that quantifies how well the global goal is achieved, and a set of local costs that characterize the influence on user experience:
\begin{align}\label{objective}
	\min_{x=(x_1,\ldots,x_N)\in \mathcal{X}}\ \ &
	F(x)= \phi(x)+\sum_{i=1}^N f_i(x).
\end{align}
Here each agent is associated with a decision variable $x_i\in\mathcal{X}_i$ where $\mathcal{X}_i\subseteq \mathbb{R}^{d_i}$ is the corresponding feasible set; the feasible set $\mathcal{X}_i$ appears naturally in many practical scenarios, and can be used to model, for instance, the range of power generation of a distributed generator.
We assume that each $\mathcal{X}_i$ is compact and convex, and has a nonempty interior in $\mathbb{R}^{d_i}$; we also assume that $0\in\operatorname{int}\mathcal{X}_i$ without loss of generality.
The joint decision variable is denoted by $x=(x_1,\ldots,x_N)$, and we also denote $\mathcal{X}=\prod_{i=1}^N\mathcal{X}_i$. Each $f_i:\mathcal{X}\rightarrow\mathbb{R}$ is the local cost function of agent $i$; the value of $f_i(x)$ only depends on the subvector $x_i$, and it is for notational purposes that we let the domain of $f_i$ be $\mathcal{X}$ instead of $\mathcal{X}_i$. The function $\phi:\mathcal{X}\rightarrow\mathbb{R}$ is the global objective. We assume that the value of the global objective function can only be observed by the aggregator, and each local cost function is only known to the associated agent. For notational simplicity, we also denote $d = \sum_{i=1}^N d_i$ and $f(x)=\sum_{i=1}^N f_i(x)$.

Next, we elaborate further details on the global objective $\phi$ and the local cost functions $f_i$.

\vspace{3pt}
\noindent{\bf Global objective function.} As mentioned before, different types of global objectives have been proposed in existing literature, which can be convex or nonconvex. In this work, we shall assume sufficient smoothness of the global objective without assuming its detailed formulation. On the other hand, we impose the following restrictions on the type of information that can be accessed about the global objective:
\begin{itemize}[leftmargin=10pt, itemindent=0pt]
\item The aggregator can only access the value of the global objective function $\phi$, and no gradient information of $\phi$ is available.
\item The global objective value can only be accessed by the aggregator but not by any of the agents.
\end{itemize}
We present one example to motivate the above restrictions: Suppose the aggregator needs to curtail the load of a distribution feeder by coordinating multiple users to meet a certain target of the total power consumption so that the safety of the distribution feeder will not be compromised. The target of the total power consumption is specified by the demand response signal $D$ sent by the grid operator, and the aggregator needs to minimize the difference between the true total power consumption of the distribution feeder and the demand response signal $D$, while also maintaining the voltage magnitudes of the buses within certain operational limits. In this case, the global objective function can be given by
\begin{subequations}
\label{eq:global_cost}
\begin{equation}
	\phi (x)=\alpha_D\cdot (p_c(x)-D)^2+\alpha_v\cdot \rho(x),
\end{equation}
where
\begin{equation}
\rho(x) = \sum_j
 \left(\max\{v_j(x)-\overline{v},0\}^2+
 \max\{\underline{v}-v_j(x),0\}^2\right).
\end{equation}
\end{subequations}
Here the mapping $p_c:\mathcal{X}\rightarrow\mathbb{R}$ maps the joint decision variable $x$ to the total power consumption of the distribution feeder measured at the substation, and $v_j:\mathcal{X}\rightarrow\mathbb{R}$ maps the joint decision $x$ to the voltage magnitude at bus $j$ in the distribution feeder. $\rho(x)$ denotes the penalty incurred when any voltage magnitude is out of the specified range $[\underline{v},\overline{v}]$. $\alpha_D$ and $\alpha_v$ are the positive linear weights of the deviation term and the penalty term. Note that the global objective function $\phi$ will in general not be available to the users, since it involves the structure of the distribution feeder as well as the confidential DR signal $D$. Moreover, the mappings $p_c$ and each $v_j$ may have an implicit or explicit relation with $x$ depending on the modeling method. For example, \cite{qin2018consensus} employs the simplified model $p_c(x)=\sum_{i=0}^{N}(1+\gamma_i)x_i$, where $\gamma_i$ is a simplified power loss-related coefficient depending on the network topology and parameters of transmission lines; however, the coefficients $\gamma_i$ in practice are hardly known, and employing this simplified model will inevitably compromise accuracy. In principle, as indicated by the AC power flow equations, $p_c$ and $v_j$ can be nonlinear and nonconvex and generally have no explicit forms. Furthermore, constructing a good mathematical model of the distribution feeder requires knowledge of the detailed topology and system parameters of the grid that are sufficiently accurate, which can be challenging when the number of buses is large and the connected devices are highly heterogeneous. In such cases, only the value of the global objective $\phi(x)$ can be observed/measured by the aggregator, and its gradient computation can be difficult. This hinders the application of traditional gradient-based optimization methods.

\vspace{3pt}
\noindent{\bf Local cost functions.}
Local costs are inevitable when users participate in DR by adjusting their load levels. Many existing works have proposed quadratic forms of local costs for different types of load resources. For example, quadratic utility functions have been widely adopted \cite{deng2015fast,qin2018consensus} for residential users of which the load resources are adjustable household appliances, and as discussed in~\cite{samadi2012advanced}, quadratic utility functions exhibit many excellent properties.
For distributed generators, \cite{nicholson1973optimum,qin2018consensus} used quadratic functions to model the revenue loss and costs incurred by adjustment of the power generation. \cite{jin2020manage,qin2018consensus} used quadratic functions to model the local costs for energy storage load.

In this paper, we do not confine the local cost functions to be quadratic or even convex; our theoretical analysis will take into account both the convex and the nonconvex settings. However, we assume that a mathematical model of the local cost $f_i$ is known (and only known) to agent $i$, which allows agent $i$ to compute the partial gradient $\nabla_i f_i(x)$ whenever the subvector $x_i$ is given.

At the end of this subsection, we introduce the notions of Lipschitz continuity and smoothness that will be used for our theoretical analysis.
\begin{definition}
Let $h:\mathcal{X}\rightarrow\mathbb{R}$ be given.
\begin{enumerate}
\item We say that $h$ is $\Lambda$-Lipschitz for some $\Lambda>0$, if for all $x,y\in\mathcal{X}$, we have
\[
|h(x)-h(y)|\leq \Lambda\|x-y\|.
\]
\item We say that $h$ is $L$-smooth for some $L>0$, if $h$ is continuously differentiable over $\mathcal{X}$, and for all $x,y\in\mathcal{X}$,
\[
\|\nabla h(x)-\nabla h(y)\|\leq L\|x-y\|.
\]
\item We say that $h$ is $(L_1,\ldots,L_d)$-coordinatewise smooth, if for each $\alpha=1,\ldots,d$,
\[
\left|\frac{\partial h(x+te_\alpha)}{\partial x_\alpha}-
\frac{\partial h(x)}{\partial x_\alpha}\right|
\leq L_\alpha|t|
\]
for all $x\in\mathcal{X}$ and all $t\in\mathbb{R}$ such that $x+te_\alpha\in\mathcal{X}$, where $e_\alpha$ is a unit vector with the $\alpha$th entry being 1.
\end{enumerate}
\end{definition}

\subsection{Preliminaries on Zeroth-Order Optimization}
In order to solve the DDR problem~\eqref{objective} with the restriction that only function value information on $\phi$ is available, we resort to derivative-free optimization approaches, particularly the zeroth-order gradient estimation technique.

Zeroth-order gradient estimation is a derivative-free optimization technique that has recently attracted researchers' attention. Existing works have shown that optimization methods based on zeroth-order gradient estimation can usually enjoy theoretical convergence guarantees that are similar to their first-order counterparts~\cite{duchi2015optimal,nesterov2017random}, and that it's relatively straightforward to adapt zeroth-order gradient estimation techniques for distributed optimization~\cite{hajinezhad2019zone,wang2022distributed,tang2023zeroth}. Given a continuously differentiable function $h:\mathbb{R}^p\rightarrow\mathbb{R}$, a commonly used zeroth-order gradient estimator for $h$ is the $2$-point gradient estimator given by
\begin{equation}
\label{grad_esti2}
G_h(x;r,z)=\frac{h(x+rz)-h(x)}{r}z.
\end{equation}
Here $z\in\mathbb{R}^p$ is a random perturbation vector whose distribution is usually chosen to be one of the following:
\begin{enumerate}
\item The standard Gaussian distribution $\mathcal{N}(0,I_p)$;
\item The uniform distribution on the sphere of radius $\sqrt{p}$, which we denote by $\mathcal{U}(\sqrt{p}\,\mathbb{S}_{p-1})$.
\end{enumerate}
The parameter $r>0$ is called the \emph{smoothing radius}, which controls the amount of perturbation in the gradient estimator. Note that to construct~\eqref{grad_esti2}, we need \emph{two} function evaluations of $h$, hence the name \emph{$2$-point gradient estimator}; in practice, these two quantities can be obtained by applying the decision variables $x$ and $x+rz$ to the system and then observe the corresponding \emph{feedback} values.

The following lemma bounds the bias of the $2$-point zeroth-order gradient estimator.

\begin{lemma}[{\cite{malik2020derivative}}]\label{lemma:zeroth-order_grad_est_bias}
Suppose $h:\mathbb{R}^p\rightarrow\mathbb{R}$ is $L$-smooth, and let $z$ be sampled from either $\mathcal{N}(0,I_p)$ or $\mathcal{U}(\sqrt{p}\,\mathbb{S}_{p-1})$. Then
\begin{align*}
\left\|\mathbb{E}\!\left[
G_h(x;r,z)-\nabla h(x)
\right]\right\|
\leq \sqrt{p}Lr.
\end{align*}
\end{lemma}

Lemma~\ref{lemma:zeroth-order_grad_est_bias} justifies that, $G_h(x;r,z)$ can serve as a stochastic gradient whose bias can be controlled by the smoothing radius $r$. Then, by plugging~\eqref{grad_esti2} into a stochastic first-order method (e.g., stochastic gradient descent or mirror descent), we obtain a zeroth-order optimization method.

%% file: algorithm_design_new.tex
\section{Algorithms} \label{section_3}

In this section, we design distributed algorithms for the DDR problem~\eqref{objective} leveraging tools from zeroth-order optimization. We shall present two zeroth-order feedback-based optimization algorithms, one called \emph{2-point Zeroth-order Feedback-based Gradient Descent} (2-ZFGD) which is based on the projected stochastic gradient descent framework, and the other called \emph{Randomized Zeroth-order Feedback-based Coordinate Descent} (RZFCD) which is based on the randomized projected coordinate descent framework. The details of the two algorithms as well as their advantages and disadvantages will be presented and discussed in the subsequent subsections.

\subsection{2-ZFGD}

Our first algorithm is based on the framework of stochastic projected gradient descent:
\[
x(k+1) = \mathcal{P}_{\mathcal{X}}\!\left[
x(k) - \eta\, g(k)
\right],
\]
where $\eta>0$ is the step size, and $g(k)$ is an estimator of the gradient $\nabla F(x(k))$. Since $\mathcal{X}$ is the Cartesian product of $\mathcal{X}_1,\ldots,\mathcal{X}_N$, we can rewrite the above iteration equivalently as
\[
x_i(k+1) = \mathcal{P}_{\mathcal{X}_i}\!\left[
x_i(k) - \eta\, g_i(k)
\right],
\]
where $x_i(k)$ is now the subvector of $x(k)$ associated with agent $i$, and each $g_i(k)$ is an estimator of the partial gradient $\nabla_i F(x(k)) = \nabla_i f_i(x(k))+\nabla_i\phi(x(k))$. Recall that in our problem setup, only zeroth-order information of the global objective $\phi$ can be accessed directly by the aggregator; moreover, in distributed demand response programs, the users may prefer not to reveal information on their local costs $f_i$ to the aggregator due to privacy issues. Meanwhile, it can be observed that $f_i(x)$ and $\nabla_i f_i(x)$ are known to agent $i$ and are not dependent on other agents' decision variables $x_j$ for $j\neq i$. Taking these considerations and observations into account, we propose the $2$-ZFGD algorithm presented in Algorithm~\ref{algorithm:2-zfgd}.


\begin{algorithm}
\caption{2-point Zeroth-order Feedback-based Gradient Descent (2-ZFGD)}\label{algorithm:2-zfgd}
	\DontPrintSemicolon
	\SetAlgoLined
 \KwIn{Number of iterations $K$, step size $\eta$, smoothing radii $(r(k))_{k\geq 0}$, shrinkage factor $\delta$}
	\For{$k\leftarrow 0$ \KwTo $K-1$}{
		Aggregator observes $\phi(x(k))$ and broadcasts it to all agents.\;
		Each agent $i$ generates $z_i(k)$ according to \eqref{projected_guassian}.\;
		Each agent applies the perturbed iterate $x_i(k)+r(k) z_i(k)$ to the system.\;
		Aggregator observes $\phi(x(k)+r(k) z(k))$ and broadcasts it to all agents.\;
		Each agent $i$ calculates
		$$
  g_i(k)=\nabla_i f_i(x(k))+\frac{\phi(x(k)+r(k) z(k))-\phi(x(k))}{r(k)}\cdot z_i(k).
		$$\;
  \vspace{-12pt}
		Each agent $i$ updates
		$$
		x_i(k+1) =\mathcal{P}_{(1-\delta)\mathcal{X}_i}[x_i(k) - \eta g_i(k)],
		$$
		and applies $x_i(k+1)$ to the system.
	}
\end{algorithm}

The design of the 2-ZFGD algorithm employs zeroth-order feedback techniques to coordinate distributed agents. Based on the gradient estimation method~\eqref{grad_esti2}, we let
\begin{align*}
g_i(k)=\nabla_i f_i(x(k))+\frac{\phi(x(k)+r(k)z(k))-\phi(x(k))}{r(k)}\cdot z_i(k),
\end{align*}
where each $z_i(k)\in\mathbb{R}^{d_i}$ is a random vector and we let $z(k)\in\mathbb{R}^d$ denote the concatenation of $z_1(k),\ldots,z_N(k)$. The probability distribution of each $z_i(k)$ needs to be designed carefully: On the one hand, it is natural to sample $z_i(k)$ from the Gaussian distribution $\mathcal{N}(0,I_{d_i})$ independently, so that the resulting joint random perturbation $z(k)$ follows the Gaussian distribution $\mathcal{N}(0,I_d)$, and we have $\mathbb{E}[g_i(k)|x_i(k)]\approx \nabla_i f_i(x(k))+\nabla_i\phi(x(k))$ by Lemma~\ref{lemma:zeroth-order_grad_est_bias}. On the other hand, the distribution $\mathcal{N}(0,I_{d_i})$ is not compactly supported, meaning that the perturbed iterate $x(k)+r(k) z(k)$ is not guaranteed to lie in the feasible set $\mathcal{X}$.
To address this issue, we adapt the technique proposed in~\cite{tang2023zeroth} and slightly modify the sampling of the perturbation $z(k)$ as follows. For each agent $i$, define
\[
S_i(x_i,r)\coloneqq
\left\{
\left.\frac{s-x_i}{r}\,\right| s\in\mathcal{X}_i
\right\},
\qquad x_i\in \operatorname{int} \mathcal{X}_i,\ \ r>0,
\]
It is obvious that
$x_i + rz_i\in\mathcal{X}_i$ for any $z_i\in S_i(x_i,r)$. We then let $z_i(k)$ be sampled by
\begin{align}
\label{projected_guassian}
z_i(k)=\mathcal{P}_{S_i(x_i(k),r(k))}[\bar{z}_i(k)], \ \ \bar{z}_i(k)\sim\mathcal{N}(0,I_{d_i}),
\end{align}
i.e., we first generate a random vector from the distribution $\mathcal{N}(0, I_{d_i})$, and then project it onto the set $S_i(x_i(k),r(k))$. We denote the distribution of $z_i(k)$ and $z(k)$ by $\mathcal{Z}_i(x_i(k),r(k))$ and $\mathcal{Z}(x(k),r(k))$ respectively.

In order for the distribution $\mathcal{Z}(x(k),r(k))$ to be close to the original Gaussian distribution $\mathcal{N}(0,I_d)$, we require that $S_i(x_i(k),r(k))$ should contain a ball with a sufficiently large radius, so that projections in~\eqref{projected_guassian} happen rarely; in this case, the statistical properties of the partial gradient estimators $g_i(k)$ will not change much, and we still have $\mathbb{E}[g_i(k)|x_i(k)]\approx \nabla f_i(x_i(k))+\nabla_i\phi(x(k))$. In order for $S_i(x_i(k),r(k))$ to satisfy this requirement, we employ the following modified version of the projected gradient descent step
$$
x_i(k+1)=\mathcal{P}_{(1-\delta)\mathcal{X}_i}[x_i(k)-\eta\, g_i(k)],
$$
where we project $x_i(k)-\eta\, g_i(k)$ onto a shrunk set $(1-\delta)\mathcal{X}_i$ for some $\delta\in(0,1)$. As shown in~\cite[Observation 3.2]{flaxman2004online}, when the shrinkage factor $\delta$ is chosen properly, the distance between $x_i(k)$ and the boundary of $\mathcal{X}_i$ will be sufficiently large, and consequently, the set $S_i(x_i(k),r(k))$ will contain a ball with a sufficiently large radius.

After having explained the critical details in the design of the 2-ZFGD algorithm, we present theoretical results on its convergence behavior. We define the following auxiliary quantities
\[
\underline{R}\coloneqq \sup\{R>0:R\mathbb{B}_{d}\subseteq\mathcal{X}\},
\qquad
\overline{R} \coloneqq
\inf\{R>0: \mathcal{X}\subseteq R\mathbb{B}_{d}\}.
\]
Since we assume that $\mathcal{X}$ is compact and $0\in \operatorname{int}\mathcal{X}$ without loss of generality, we have $0<\underline{R}\leq\overline{R}<+\infty$.

We first provide the performance guarantees of 2-ZFGD for the convex case, summarized in the following theorem.

\begin{theorem}\label{theorem:2ZFGD_convergence}
Suppose $F$ is convex, $\Lambda_F$-Lipschitz and $L_F$-smooth, and $\phi$ is $\Lambda_\phi$-Lipschitz and $L_\phi$-smooth over $\mathcal{X}$. Without loss of generality, we let $\Lambda_\phi\leq\Lambda_F$ and $L_\phi\leq L_F$. Let $x^\ast$ be a minimizer of $F(x)$ over $x\in\mathcal{X}$. Then, for any sufficiently small $\epsilon>0$, if we choose the algorithmic parameters to satisfy
\[
\delta\leq \frac{\epsilon}{5\Lambda_F\left(\overline{R}+\Lambda_\phi/(2L_F d)\right)},
\qquad
\eta\leq \frac{1}{2(d+5)}\min\left\{
\frac{\epsilon}{5\Lambda_\phi^2},\frac{1}{L_F}
\right\},
\qquad K\geq \frac{10\overline{R}^2}{\eta\epsilon},
\]
\[
\sum\nolimits_{k=0}^\infty r(k) \leq 2\sqrt{d}\overline{R},
\qquad
\sum\nolimits_{k=0}^\infty r(k)^2\leq 
\frac{4\overline{R}^2}{d+5},
\]
and
\[
r(k) \leq\frac{\delta\underline{R}}{2\sqrt{\frac{d}{2}+4\ln\frac{8\overline{R}}{\underline{R}}+2\ln\frac{d}{\delta^3}}},
\qquad\forall k=0,\ldots,K-1,
\]
it can be guaranteed that the sequence $\{x(k)\}_{k=0}^K$ generated by 2-ZFGD satisfies
\[
\min_{1\leq k\leq K}\mathbb{E}\!\left[F(x(k))-F(x^\ast)\right]\leq\epsilon.
\]
\end{theorem}
The proof of Theorem~\ref{theorem:2ZFGD_convergence} is postponed to Appendix~\ref{proof_of_theorem1}. As a corollary, we have the following complexity bound of 2-ZFGD for the convex case.
\begin{corollary}\label{corollary:2ZFGD_complexity}
Suppose the functions $F$ and $\phi$ satisfy the conditions in Theorem~\ref{theorem:2ZFGD_convergence}. Let $\epsilon>0$ be arbitrary. Then the number of iterations needed to achieve
\[
\min_{1\leq k\leq K} \mathbb{E}\!\left[F(x(k))-F(x^\ast)\right]\leq\epsilon
\]
for 2-ZFGD can be upper bounded by
$
\mathcal{O}\!\left(d/\epsilon^2\right)$.
\end{corollary}

To analyze the performance of 2-ZFGD for the nonconvex case, we introduce the following stationarity measure
\[
\mathfrak{g}(x;M)\coloneqq
M\!\left(x - \mathcal{P}_{\mathcal{X}}\!\left[
x -  \frac{1}{M}\nabla F(x)
\right]\right)
\]
for any $x\in\mathcal{X}$ and $M>0$. The following lemma suggests that we may employ $\|\mathfrak{g}(x;M)\|^2$ to quantify how close $x$ is to being a stationary point of $F$ over $\mathcal{X}$.
\begin{lemma}
\label{lemma:gradient_mapping_properties}
Suppose $F:\mathcal{X}\rightarrow\mathbb{R}$ is continuously differentiable, and let $M>0$ be arbitrary. We have
\begin{enumerate}
\item $x\mapsto \|\mathfrak{g}(x;M)\|^2$ is a continuous function over $x\in\mathcal{X}$.
\item Given $x^\ast\in\mathcal{X}$, we have $\|\mathfrak{g}(x^\ast;M)\|^2=0$ if and only if
\[
\left.\frac{d}{dt} F(x^\ast+t(x-x^\ast))\right|_{t=0}\geq 0
\]
for all $x\in\mathcal{X}$.
\end{enumerate}
\end{lemma}
The results in Lemma~\ref{lemma:gradient_mapping_properties} are standard in optimization theory and we omit the proofs here. We also mention that $\|\mathfrak{g}(x;M)\|^2$ has been adopted for measuring distance from stationarity for constrained nonconvex smooth problems in existing literature~\cite{nesterov2013gradient,scutari2019distributed,he2022zeroth}.

The following theorem provides performance guarantees of 2-ZFGD for the nonconvex case, whose proof is postponed to Appendix~\ref{proof_of_theorem2}.

\begin{theorem}
\label{theorem:2ZFGD_convergence_nonconvex}
Suppose $F$ is $L_F$-smooth, and $\phi$ is $\Lambda_\phi$-Lipschitz and $L_\phi$-smooth over $\mathcal{X}$. Without loss of generality, we let $L_\phi\leq L_F$. Let $F^\ast\coloneqq \min_{x\in\mathcal{X}} F(x)$. Then, for any sufficiently small $\epsilon>0$, if we choose the algorithmic parameters to satisfy
\[
\delta\leq \frac{\sqrt{\epsilon}}{5L_F(\overline{R}+\Lambda_\phi/(2L_Fd))},
\quad
\eta \leq \frac{1}{L_F(d+5)}\min\left\{
\frac{\epsilon}{30\Lambda_\phi^2},1
\right\},
\quad
K\geq\frac{15(F(x(0))-F^\ast)}{\eta\epsilon}
\]
and
\[
\sum_{k=0}^\infty r(k)^2
\leq\frac{F(x(0))-F^\ast}{L_\phi(d+6)},
\qquad
r(k)\leq \frac{\delta\underline{R}}{
2\sqrt{\frac{d}{2}+4\ln\frac{8\overline{R}}{\underline{R}}+\ln\frac{d^3}{\delta^7}}
},
\]
it can be guaranteed that the sequence $\{x(k)\}_{k=0}^K$ generated by 2-ZFGD satisfies
\[
\min_{0\leq k\leq K-1}\mathbb{E}\!\left[\|\mathfrak{g}(x(k);L_F)\|^2\right]\leq\epsilon.
\]
\end{theorem}
\begin{corollary}\label{corollary:2ZFGD_complexity_nonconvex}
Suppose the functions $F$ and $\phi$ satisfy the conditions in Theorem~\ref{theorem:2ZFGD_convergence_nonconvex}. Let $\epsilon>0$ be arbitrary. Then the number of iterations needed to achieve
\[
\min_{0\leq k\leq K-1} 
\mathbb{E}\!\left[
\|\mathfrak{g}(x(k);L_F)\|^2
\right]\leq\epsilon
\]
for 2-ZFGD can be upper bounded by
$
\mathcal{O}\!\left(d/\epsilon^2\right)$.
\end{corollary}

Theorems~\ref{theorem:2ZFGD_convergence}--\ref{theorem:2ZFGD_convergence_nonconvex} and Corollaries~\ref{corollary:2ZFGD_complexity}--\ref{corollary:2ZFGD_complexity_nonconvex} establish the convergence guarantees and iteration complexity bounds for 2-ZFGD. Particularly, since each iteration of 2-ZFGD requires accessing two values of $\phi$, the bounds in Corollaries~\ref{corollary:2ZFGD_complexity}--\ref{corollary:2ZFGD_complexity_nonconvex} are also oracle complexity bounds in the sense that they bound the number of zeroth-order queries needed to achieve certain degree of optimality/stationarity for 2-ZFGD. These complexity bounds provide quantitative characterizations of the efficiency of 2-ZFGD.

We notice that the oracle complexity bound $\mathcal{O}(d/\epsilon^2)$ suggests that 2-ZFGD may still have room for improvement. Specifically, recalling that the oracle complexity of zeroth-order optimization for \emph{unconstrained} smooth problems can be upper bounded by $\mathcal{O}(d/\epsilon)$ \cite{nesterov2017random}, we can clearly see a gap between the bound $\mathcal{O}(d/\epsilon^2)$ of 2-ZFGD and the bound $\mathcal{O}(d/\epsilon)$ in terms of the dependence on $\epsilon$. This gap does not occur in first-order methods, as the oracle complexities of the first-order deterministic (projected) gradient descent are $O(1/\epsilon)$ \cite{nesterov2018lectures} regardless of whether the problem is constrained or not. Theoretical analysis reveals that this gap is not due to the distributed setting but results from the particular form of the variance of the zeroth-order gradient estimator: Given a smooth function $h:\mathbb{R}^p\rightarrow\mathbb{R}$, it can be derived that
\begin{equation}\label{eq:variance_two_point_zero_r}
\lim_{r\downarrow 0}\mathbb{E}\!\left[
\|G_h(x;r,z)-\nabla h(x)\|^2
\right]
=
(d+1)\|\nabla h(x)\|^2,
\end{equation}
(see Appendix~\ref{appendix:proof_variance_two_point_zero_r}). For unconstrained optimization $\min_{x\in\mathbb{R}^p} h(x)$, as we approach an optimal point $x^\ast$, the gradient will converge to zero. Consequently, as long as the smoothing radii are chosen appropriately, the variance of $G_h(x(k);r(k),z(k))$ will be negligible and the convergence of the zeroth-order iteration $x(k+1) = x(k)-\eta\,G_h(x(k);r,z(k))$ resembles deterministic gradient descent, leading to a complexity bound proportional to $\epsilon^{-1}$. However, for the constrained problem $\min_{x\in\mathcal{X}} h(x)$, the optimal point $x^\ast$ may lie on the boundary of $\mathcal{X}$ with a nonzero gradient, meaning that the variance of 2-point gradient estimation will be approximated by $(d+1)\|\nabla h(x^\ast)\|^2>0$ as we approach $x^\ast$. Consequently, the convergence of the iteration $x(k+1)=\mathcal{P}_{\mathcal{X}}[x(k)-\eta\,G_h(x(k);r,z(k))]$ resembles stochastic projected gradient descent, and the complexity bound is proportional to $\epsilon^{-2}$ which is strictly inferior. In Section~\ref{section:case_studies}, we will provide experimental results on certain numerical test cases for 2-ZFGD, showing that the convergence of 2-ZFGD can indeed be slow and may not meet the requirement on efficiency for practical applications.

The gap in the oracle complexity and the slow convergence of 2-ZFGD naturally raises the following interesting and important question: \emph{Can we further improve the convergence behavior of the distributed zeroth-order optimization algorithm and close the aforementioned gap in the oracle complexity}? We shall see in the next subsection that the answer to this question is positive, provided that we impose further assumptions on the feasible set $\mathcal{X}$.

\subsection{RZFCD}
To solve the above problem of slow convergence, we propose another distributed zeroth-order optimization method called \emph{Randomized Zeroth-order Feedback-based Coordinate Descent} (RZFCD). 
We impose the critical assumption in the design of RZFCD that each feasible set $\mathcal{X}_i$ is of the form $\{x_i\in\mathbb{R}^{d_i}:l_i\leq x_i\leq u_i\}$ for some $l_i,u_i\in\mathbb{R}^{d_i}$, i.e., each $\mathcal{X}_i$ is a multi-dimensional box. For simplicity of exposition, we let $d_i=1$ for each $i$, but the extension to the situations with $d_i\geq 1$ is straightforward.

\begin{algorithm}[tb]
\caption{Randomized Zeroth-order Feedback-based Coordinate Descent (RZFCD)}\label{algorithm:rzfcd}
\DontPrintSemicolon
\SetAlgoLined
\KwIn{Number of iterations $K$, step sizes $\eta_\alpha$ and smoothing radii $(r_\alpha(k))_{k\geq 0}$ for each $\alpha$}
\For{$k\leftarrow 0$ \KwTo $K-1$}{

The aggregator samples $\alpha(k)$ uniformly from $\{1,\ldots,d\}$.\;

The aggregator sends $\phi(x(k))$ to agent $\alpha(k)$.\;

\uIf{$x_{\alpha(k)}(k) + r_{\alpha(k)}(k) > u_{\alpha(k)}$}{
Agent $\alpha(k)$ sets $z_{\alpha(k)}(k) = -1$.\;
}
\uElseIf{$x_{\alpha(k)}(k) - r_{\alpha(k)}(k) < l_{\alpha(k)}$}{
Agent $\alpha(k)$ sets $z_{\alpha(k)}(k) = 1$.\;
}
\Else{
Agent $\alpha(k)$ samples $z_{\alpha(k)}(k)$ uniformly from $\{+1,-1\}$.\;
}
$z_\beta(k) \leftarrow 0$ for $\beta\neq \alpha(k)$.\;

Agent $\alpha(k)$ applies $x_{\alpha(k)}(k) + r_{\alpha(k)}(k) z_{\alpha(k)}(k)$ to the system, while other agents keep their decision variables unchanged.\;

The aggregator observes $\phi(x(k)+r_{\alpha(k)}(k) z(k))$ and sends it to agent $\alpha(k)$.\;

Agent $\alpha(k)$ calculates
$$
g_{\alpha(k)}(k) = \frac{\partial f_{\alpha(k)}(x(k))}{\partial x_{\alpha(k)}} + \frac{\phi(x(k)+r_{\alpha(k)}(k) z(k))-\phi(x(k))}{r_{\alpha(k)}(k)}z_{\alpha(k)}(k).
$$\;\vspace{-16pt}
Each agent $\beta$ updates
$$
x_\beta(k+1) =\left\{
\begin{aligned}
& \mathcal{P}_{[l_\beta,u_\beta]}[x_\beta(k) - \eta_\beta g(k)], & & \beta=\alpha(k), \\
& x_\beta(k), & & \beta\neq \alpha(k),
\end{aligned}
\right.
$$
and applies $x_\beta(k+1)$ to the system.\;
}

\end{algorithm}

The details of the RZFCD algorithm are presented in Algorithm~\ref{algorithm:rzfcd}. The key difference between RZFCD and 2-ZFGD is that the design of RZFCD employs randomized coordinate descent as the framework. For each iteration $k$, the aggregator first randomly selects an agent $\alpha(k)$ uniformly from the set $\{1,2,\ldots,d\}$. We then fix all other entries of $x(k)$ and consider optimizing only over the $\alpha(k)$'th entry. The estimation of the partial gradient of $F$ with respect to $x_{\alpha(k)}$ will be carried out by the corresponding agent $\alpha(k)$ together with the aggregator, which is given by
\begin{align*}
g_{\alpha(k)}(k)=\frac{\partial f_{\alpha(k)}(x(k))}{\partial x_{\alpha(k)}}+\frac{\phi(x(k)+r_{\alpha(k)}(k)z(k))-\phi(x(k))}{r_{\alpha(k)}(k)}z_{\alpha(k)}(k),
\end{align*}
i.e., only the partial gradient of one dimension is estimated. The error of gradient estimation is bounded by controlling the smoothing radius $r_{\alpha(k)}(k)$. Here we let $r_\beta(k)$ denote the sequence of smoothing radii used for the $\beta$'th coordinate, and we allow different coordinates to employ different sequences of smoothing radii. The random perturbation $z(k)\in\mathbb{R}^d$ is a vector with only the $\alpha(k)$'th entry being nonzero, with $z_{\alpha(k)}(k)$ given by
$$
z_{\alpha(k)}(k)\left\{
\begin{aligned}
& =1, & & x_{\alpha(k)}(k) - r_{\alpha(k)}(k) < l_{\alpha(k)},\\
& =-1, & & x_{j(k)}(k) + r_{\alpha(k)}(k) > u_{\alpha(k)},\\
& \sim\mathcal{U}\{1,-1\}, & & \text{otherwise}.
\end{aligned}
\right.
$$
This sampling strategy for the random perturbation $z(k)$ is different from 2-ZFGD, and is based on the uniform distribution on the sphere $\sqrt{p}\,\mathbb{S}_{p-1}=\{1,-1\}$ rather than the Gaussian distribution (we have $p=1$ since all entries but $x_{\alpha(k)}$ are fixed); we have also made slight modifications to ensure that $x(k)+r_{\alpha(k)}(k)z(k)\in\mathcal{X}$, which is simpler than 2-ZFGD as each $\mathcal{X}_i$ is assumed to be a box.
We then apply the projected coordinate descent step to update $x_{\alpha(k)}$:
\[
x_{\alpha(k)}(k+1) = \mathcal{P}_{[l_{\alpha(k)},u_{\alpha(k)}]}
\!\left[
x_{\alpha(k)}(k) - \eta_{\alpha(k)} g_{\alpha(k)}(k)
\right].
\]
The quantities $\eta_\beta$ for each $\beta=1,\ldots,d$ are the step sizes, and we allow them to differ when different entries of the decision variable are updated. 

After having explained the rationale of RZFCD, we present theoretical convergence guarantees for RZFCD. The proofs of these theoretical results will be given in Section~\ref{sec:analysis}.

For the convex case and $\{x(k)\}_{k\geq 0}$ derived from Algorithm~\ref{algorithm:rzfcd}, we have the following theorem.
\begin{theorem}\label{theorem:RZFCD_convex}
Suppose that the function $F$ is convex and
$(L_{F,1},\ldots,L_{F,d})$-coordinatewise smooth, and $x^\ast$ is a minimizer of $F(x)$ over $x\in\mathcal{X}$. Further, suppose $\phi$ is $(L_{\phi,1},\ldots,L_{\phi,d})$-coordinatewise smooth. Let the step sizes satisfy $\eta_\beta L_{F,\beta}\leq 1$ for all $\beta=1,\ldots,d$ and the smoothing radii satisfy
\[
\sum_{\beta=1}^d\sum_{k=1}^\infty r_\beta(k)
<+\infty.
\]
Then, for the sequence $\{x(k)\}_{k\geq 0}$ generated by RZFCD, we have
\[
\min_{0\leq k\leq K}\mathbb{E}\!\left[
F(x(k)) - F(x^\ast)
\right]
\leq \mathcal{O}\!\left(\frac{d}{K}\right),
\]
\end{theorem}
\begin{corollary}
\label{corollary:RZFCD_complexity_convex}
Suppose the functions $F$ and $\phi$ satisfy the conditions in Theorem~\ref{theorem:RZFCD_convex}. Let $\epsilon>0$ be arbitrary. Then the number of iterations needed to achieve
\[
\min_{0\leq k\leq K}\mathbb{E}\!\left[
F(x(k)) - F(x^\ast)
\right]\leq\epsilon
\]
for RZFCD can be upper bounded by
$\mathcal{O}\!\left(d/\epsilon\right)$.
\end{corollary}

The following theorem summarizes the convergence results for RZFCD for the nonconvex case.
\begin{theorem}\label{theorem:RZFCD_nonconvex}
Suppose $F$ is $(L_{F,1},\ldots,L_{F,d})$-coordinatewise smooth, and $\phi$ is $(L_{\phi,1},\ldots,L_{\phi,d})$-coordinatewise smooth. Without loss of generality we let $L_{F,\beta}\geq L_{\phi,\beta}$ for all $\beta$. Let the step sizes satisfy $\eta_\beta L_{F,\beta}\leq 1$ for all $\beta=1,\ldots,d$, and let the smoothing radii satisfy
\[
\sum_{\beta=1}^d\sum_{k=1}^\infty r_\beta(k)<+\infty.
\]
Then, for the sequence $\{x(k)\}_{k\geq 0}$ generated by RZFCD, we have
\[
\min_{0\leq k\leq K-1}
\mathbb{E}\!\left[\|\mathfrak{g}(x(k);\underline{L}_{F})\|^2\right]
\leq \mathcal{O}\!\left(\frac{d}{K}\right),
\]
where $\underline{L}_F=\min_{1\leq \beta\leq d} L_{F,\beta}$.
\end{theorem}
\begin{corollary}\label{corollary:RZFCD_complexity_nonconvex}
Suppose the functions $F$ and $\phi$ satisfy the conditions in Theorem~\ref{theorem:RZFCD_nonconvex}. Let $\epsilon>0$ be arbitrary. Then the number of iterations needed to achieve
\[
\min_{0\leq k\leq K-1}
\mathbb{E}\!\left[\|\mathfrak{g}(x(k);\underline{L}_{F})\|^2\right]
\leq\epsilon
\]
for RZFCD can be upper bounded by
$\mathcal{O}\!\left(d/\epsilon\right)$.
\end{corollary}

We provide several discussions about the above theoretical results, particularly on the comparison of complexity bounds and convergence conditions.

\begin{enumerate}
\item \textbf{Comparison of complexity 
bounds with 2-ZFGD}. 
By comparing Corollaries~\ref{corollary:RZFCD_complexity_convex}--\ref{corollary:RZFCD_complexity_nonconvex} with Corollaries~\ref{corollary:2ZFGD_complexity}--\ref{corollary:2ZFGD_complexity_nonconvex}, we clearly see that the complexity bounds of RZFCD are superior to those of 2-ZFGD for both the convex and the nonconvex cases. We shall later see that these theoretical implications accord with the numerical results presented in Section~\ref{section:case_studies}. On the other hand, RZFCD requires that the feasible set $\mathcal{X}$ is a compact box, while 2-ZFGD only requires that $\mathcal{X}$ is a compact convex set, indicating that 2-ZFGD may have wider applicability than RZFCD.

\item \textbf{Comparison of complexity bounds with unconstrained optimization}. Recalling that the complexity of zeroth-order optimization for deterministic unconstrained smooth problems can be bounded by $\mathcal{O}(d/\epsilon)$ \cite{nesterov2017random}, we see that RZFCD is able to close the gap between constrained and unconstrained smooth optimization that was mentioned in the previous subsection when the feasible set is a box.

Interestingly, the existing literature does not seem to have paid enough attention to this gap between zeroth-order constrained and unconstrained optimization. To the best of our knowledge, without exploiting acceleration or variance reduction techniques, the best-known complexity bound of two-point zeroth-order methods for deterministic constrained smooth convex optimization prior to this work is given by $\mathcal{O}(d/\epsilon^2)$~\cite{duchi2015optimal}; regarding zeroth-order optimization for constrained smooth nonconvex optimization, the only relevant work that we are aware of is~\cite{he2022zeroth} that adopted a cyclic block coordinate descent approach, leading to a complexity bound proportional to $\epsilon^{-1}$ but exponential in the number of blocks.\footnote{
We also note that there are existing works on zeroth-order Frank-Wolfe methods~\cite{sahu2019towards} that achieve the $\mathcal{O}(d/\epsilon)$ complexity bound. However, these methods require a linear minimization oracle (LMO) rather than a projection oracle, which is different from our setting. In addition, these methods need $\Omega(d)$ function values to construct one gradient estimator, which can be inefficient when $d$ is large.
}  We believe that our design and analysis of RZFCD will be of independent interest to researchers in the area of general zeroth-order optimization, and can provide important insight on how to close this gap for more general settings.

\item \textbf{Coordinatewise smoothness.} In Theorems~\ref{theorem:RZFCD_convex}--\ref{theorem:RZFCD_nonconvex}, we impose \emph{coordinatewise smoothness} on the functions $F$ and $\phi$, rather than ordinary smoothness; similar conditions have been employed in~\cite{nesterov2012efficiency} for analyzing first-order coordinate descent algorithms. It's not hard to see that $L$-smoothness implies $(L,\ldots,L)$-coordinatewise smoothness, while $(L,\ldots,L)$-coordinatewise smoothness only implies $\sqrt{d}L$-smoothness. In many situations, an $L$-smooth objective function may be $(L_1,\ldots,L_d)$-coordinatewise smooth with $L_\alpha$ much smaller than $L$ for all $\alpha\in \{1,\ldots,d\}$. Considering that the step sizes of RZFCD are chosen according to $\eta_\alpha\leq 1/L_\alpha$, we see that imposing coordinatewise smoothness on the objective functions allows larger step sizes, which may lead to faster convergence.


\end{enumerate}

\begin{remark}
The theoretical analysis in this paper assumes that the values of the function $\phi$ can be accessed accurately without being corrupted by noise or error. This assumption provides convenience for theoretical analysis but is only an approximation to practical situations. Apart from the limited precision of numerical computation, noise and error from the sensors may also render the obtained values of $\phi$ inaccurate. We expect that when the zeroth-order information has relatively large noise/error, the choice of the smoothing radii $r(k)$ needs to be more conservative, which may lead to slower convergence. Detailed analysis of our proposed algorithms in the presence of noise and/or error is beyond the scope of this paper and will be an interesting future direction.
\end{remark}

\begin{remark}
Note that in the two proposed algorithms, the aggregator does not need to know or collect any information on the utility functions $f_i$ of the agents; the local gradient $\nabla_i f_i(x(k))$ is only computed locally by each agent and will not be uploaded to the aggregator. This feature helps preserve the agents' privacy during the optimization procedure, which will be important for the adoption of distributed demand response programs in practical scenarios. It will be interesting to investigate whether our algorithms will enjoy theoretically guaranteed and quantified degrees of privacy preservation from the perspective of, e.g., differential privacy, but we leave it to future work.
\end{remark}

%% file: algorithm_analysis_new.tex
\section{Analysis of RZFCD}
\label{sec:analysis}

Define the filtration $\mathcal{F}_k = \sigma(x(0),\alpha(0),x(1),\alpha(1),\ldots,x(k))$. For notational simplicity, we denote $r(k)\coloneqq r_{\alpha(k)}(k)$.

We first derive some auxiliary results that will be used for subsequent analysis. Note that the identity
$$
x_{\alpha(k)}(k+1) = \mathcal{P}_{[l_{\alpha(k)},u_{\alpha(k)}]}
\!\left[
x_{\alpha(k)}(k) - \eta_{\alpha(k)}\,g_{\alpha(k)}(k)
\right]
$$
implies
\begin{equation*}
(y-x_{\alpha(k)}(k+1))\cdot
\left(
x_{\alpha(k)}(k) - \eta_{\alpha(k)}\,g_{\alpha(k)}(k)
-x_{\alpha(k)}(k+1)
\right)\leq 0
\end{equation*}
for any $y\in[l_{\alpha(k)},u_{\alpha(k)}]$. Particularly,
\begin{equation}\label{eq:proj_intermediate_1}
(x_{\alpha(k)}(k)-x_{\alpha(k)}(k+1))\cdot
\left(
x_{\alpha(k)}(k) - \eta_{\alpha(k)}\,g_{\alpha(k)}(k)
-x_{\alpha(k)}(k+1)
\right)\leq 0
\end{equation}
and, if $x^\ast\in\mathcal{X}$ is a locally optimal point of $F$,
\begin{equation}\label{eq:proj_intermediate_2}
(x^\ast_{\alpha(k)}-x_{\alpha(k)}(k+1))\cdot
\left(
x_{\alpha(k)}(k) - \eta_{\alpha(k)}\,g_{\alpha(k)}(k)
-x_{\alpha(k)}(k+1)
\right)\leq 0.
\end{equation}

Our analysis of RZFCD will be based on the following lemma that characterizes how well the zeroth-order estimators approximate the true partial derivatives.
\begin{lemma}\label{lemma:RZFCD_zeroth-order_surrogate_func}
For each $k=1,2,\ldots$, we have
\begin{equation}\label{eq:RZFCD_grad_diff}
\left|
g_{\alpha(k)}(k)
-\frac{\partial F(x(k))}{\partial x_{\alpha(k)}}
\right|
\leq \frac{1}{2}L_{\phi,\alpha(k)}r(k).
\end{equation}
\end{lemma}
\begin{proof}
We have
\begin{align*}
\left|g_{\alpha(k)}(k) -
\frac{\partial F(x(k))}{\partial x_{\alpha(k)}}\right|
=\ &
\left|\frac{\phi(x(k)+r(k)z(k))-\phi(x(k))}{r(k)}z_{\alpha(k)}(k)
-\frac{\partial\phi(x(k))}{\partial x_{\alpha(k)}}\right| \\
=\ &
\left|\int_0^1
\left(\frac{\partial\phi(x(k)+s\cdot r(k)z(k))}{\partial x_{\alpha(k)}}
-\frac{\partial\phi(x(k))}{\partial x_{\alpha(k)}}
\right)\,ds\right| \\
\leq\ &
\int_0^1
\left|\frac{\partial\phi(x(k)+s\cdot r(k)z(k))}{\partial x_{\alpha(k)}}
-\frac{\partial\phi(x(k))}{\partial x_{\alpha(k)}}
\right|\,ds \\
\leq\ &
\int_0^1 L_{\phi,\alpha(k)}\cdot |s\cdot r(k)z_{\alpha(k)}(k)|\,ds
\leq \frac{1}{2}L_{\phi,\alpha(k)}r(k).
\qedhere
\end{align*}
\end{proof}

We now apply the coordinatewise smoothness of $F$ to get
\begin{align*}
& F(x(k+1))-F(x(k)) \\
\leq\ &
\frac{\partial F(x(k))}{\partial x_{\alpha(k)}}
(x_{\alpha(k)}(k+1)-x_{\alpha(k)}(k))
+\frac{L_{F,\alpha(k)}}{2}\left|x_{\alpha(k)}(k+1)-x_{\alpha(k)}(k)\right|^2 \\
\leq\ &
g_{\alpha(k)}(k)\,(x_{\alpha(k)}(k+1)-x_{\alpha(k)}(k))
+\frac{L_{F,\alpha(k)}}{2}\left|x_{\alpha(k)}(k+1)-x_{\alpha(k)}(k)\right|^2 \\
& + \left|\frac{\partial F(x(k))}{\partial x_{\alpha(k)}}
-g_{\alpha(k))}(k)\right|
\cdot \left|x_{\alpha(k)}(k+1)-x_{\alpha(k)}(k)\right|,
\end{align*}
which, combined with~\eqref{eq:RZFCD_grad_diff}, leads to
\begin{equation}\label{eq:RZFCD_smoothness_intermediate_1} 
\begin{aligned}
& F(x(k+1))-F(x(k)) \\
\leq\ &
g_{\alpha(k)}(k)\,(x_{\alpha(k)}(k+1)-x_{\alpha(k)}(k))
+\frac{L_{F,\alpha(k)}}{2}\left|x_{\alpha(k)}(k+1)-x_{\alpha(k)}(k)\right|^2 \\
& + \frac{1}{2} L_{\phi,\alpha(k)} r(k)(u_{\alpha(k)}-l_{\alpha(k)}).
\end{aligned}
\end{equation}

\subsection{The Convex Case}
In this subsection, we assume that each $f_i$ and $\phi$ are convex functions on $\mathcal{X}$.

We start our analysis by noting that the inequality~\eqref{eq:proj_intermediate_2} implies
\begin{align*}
& g_{\alpha(k)}(k)
\cdot(x_{\alpha(k)}(k+1)-x_{\alpha(k)}(k))
+\frac{1}{2\eta_{\alpha(k)}}|x_{\alpha(k)}(k+1)-x_{\alpha(k)}(k)|^2 \\
\leq\ & g_{\alpha(k)}(k)
\cdot(x^\ast_{\alpha(k)}-x_{\alpha(k)}(k))
+\frac{1}{2\eta_{\alpha(k)}}(x_{\alpha(k)}(k+1)-x_{\alpha(k)}(k))(2x^\ast_{\alpha(k)}-x_{\alpha(k)}(k)-x_{\alpha(k)}(k+1)) \\
=\ &
g_{\alpha(k)}(k)
\cdot(x^\ast_{\alpha(k)}-x_{\alpha(k)}(k))
+\frac{1}{2\eta_{\alpha(k)}}\left(|x^\ast_{\alpha(k)}-x_{\alpha(k)}(k)|^2
-|x^\ast_{\alpha(k)}-x_{\alpha(k)}(k+1)|^2\right).
\end{align*}
By combining it with~\eqref{eq:RZFCD_smoothness_intermediate_1} and using the condition $\eta_\beta L_{F,\beta}\leq 1$ for all $\beta$, we can get
\begin{align*}
& F(x(k+1))-F(x(k)) \\
\leq\ &
g_{\alpha(k)}(k)\,(x^\ast_{\alpha(k)}-x_{\alpha(k)}(k))
+\frac{1}{2\eta_{\alpha(k)}}\left(|x^\ast_{\alpha(k)}-x_{\alpha(k)}(k)|^2
-|x^\ast_{\alpha(k)}-x_{\alpha(k)}(k+1)|^2\right) \\
& + \frac{1}{2}L_{\phi,\alpha(k)}r(k)(u_{\alpha(k)}-l_{\alpha(k)}).
\end{align*}
By taking the expectation conditioned on $\mathcal{F}_k$, we see that
\begin{equation}\label{eq:RZFCD_convex_intermediate1}
\begin{aligned}
& \mathbb{E}\!\left[F(x(k+1))-F(x(k))\mid\mathcal{F}_k\right] \\
\leq\ &
\mathbb{E}\!\left[\left.g_{\alpha(k)}(k)\,(x^\ast_{\alpha(k)}-x_{\alpha(k)}(k))
+\frac{1}{2d}\sum_{\beta=1}^d
\left(\frac{|x^\ast_{\beta}-x_{\beta}(k)|^2}{\eta_\beta}
-\frac{|x^\ast_{\beta}-x_{\beta}(k+1)|^2}{\eta_\beta}\right)\right|\mathcal{F}_k\right] \\
&
+ \frac{1}{2d}\sum_{\beta=1}^d
L_{\phi,\beta}(u_\beta-l_\beta) r_\beta(k).
\end{aligned}
\end{equation}
To bound $g_{\alpha(k)}(k)\,(x^\ast_{\alpha(k)}-x_{\alpha(k)}(k))$, we first note that
\begin{align*}
& g_{\alpha(k)}(k)\,(x^\ast_{\alpha(k)}-x_{\alpha(k)}(k)) \\
\leq \ &
\frac{\partial F(x(k))}{\partial x_{\alpha(k)}}
(x^\ast_{\alpha(k)}-x_{\alpha(k)}(k))
+
\left|\frac{\partial F(x(k))}{\partial x_{\alpha(k)}}
-g_{\alpha(k)}(k)\right|
\left|x_{\alpha(k)}^\ast-x_{\alpha(k)}(k)\right| \\
\leq\ &
\frac{\partial F(x(k))}{\partial x_{\alpha(k)}}
(x^\ast_{\alpha(k)}-x_{\alpha(k)}(k))
+
\frac{1}{2}L_{\phi,\alpha(k)}r(k)\cdot(u_{\alpha(k)}-l_{\alpha(k)}),
\end{align*}
where we used~\eqref{eq:RZFCD_grad_diff} in the last step. Then, by taking the expectation conditioned on $\mathcal{F}_k$, we get
\begin{align*}
& \mathbb{E}\!\left[\left.g_{\alpha(k)}(k)\,(x^\ast_{\alpha(k)}-x_{\alpha(k)}(k))\,\right|\mathcal{F}_k\right] \\
\leq\ &
\frac{1}{d}\sum_{\beta=1}^d
\left(\frac{\partial F(x(k))}{\partial x_\beta}(x^\ast_\beta-x_\beta(k))
+\frac{1}{2}L_{\phi,\beta}(u_\beta-l_\beta)r_\beta(k)\right) \\
=\ &
\frac{1}{d}\langle \nabla F(x(k)),x^\ast-x(k)\rangle
+\frac{1}{2d}\sum_{\beta=1}^d L_{\phi,\beta}(u_\beta-l_\beta)r_\beta(k) \\
\leq\ &
\frac{1}{d}(F(x^\ast)-F(x(k)))
+\frac{1}{2d}\sum_{\beta=1}^d L_{\phi,\beta}(u_\beta-l_\beta)r_\beta(k),
\end{align*}
where we used the convexity of $F$ in the last step. By plugging this bound into~\eqref{eq:RZFCD_convex_intermediate1}, we can obtain
\begin{align*}
& \mathbb{E}\!\left[F(x(k+1))-F(x(k))\mid \mathcal{F}_k\right] \\
\leq\ &
\frac{1}{d}(F(x^\ast)-F(x(k)))
+\frac{1}{d}\sum_{\beta=1}^d
L_{\phi,\beta}
(u_\beta-l_\beta)r_\beta(k) \\
&
+\frac{1}{2d}\sum_{\beta=1}^d\mathbb{E}\!\left[
\left.
\frac{|x^\ast_{\beta}-x_{\beta}(k)|^2}{\eta_\beta}
-\frac{|x^\ast_{\beta}-x_{\beta}(k+1)|^2}{\eta_\beta}\,\right|\mathcal{F}_k\right].
\end{align*}
We can now take the total expectation and the telescoping sum to get
\begin{align*}
& \mathbb{E}[F(x(K))-F(x(0))] \\
\leq\ &
\frac{1}{d}\,\mathbb{E}\!\left[\sum_{k=0}^{K-1}\left(F(x^\ast)-F(x(k))\right)\right]
+\frac{1}{d}\sum_{\beta=1}^d L_{\phi,\beta}(u_\beta-l_\beta)\sum_{k=0}^{K-1}r_\beta(k)
+\frac{1}{2d}\sum_{\beta=1}^d \frac{|x^\ast_\beta-x_\beta(0)|^2}{\eta_\beta}.
\end{align*}
Finally, observe that
\[
\min_{0\leq k\leq K}
\mathbb{E}\!\left[F(x(k))-F(x^\ast)\right]
\leq \frac{d}{K+d}(F(x(K))-F(x^\ast)) + \frac{1}{K+d}\sum_{k=0}^{K-1}(F(x(k))-F(x^\ast)),
\]
and we obtain
\begin{align*}
& \min_{0\leq k\leq K}\mathbb{E}[F(x(k))-F(x^\ast)] \\
\leq\ &
\frac{d}{K+d}
\left(
F(x(0))-F(x^\ast)
+ \frac{1}{d}\sum_{\beta=1}^d L_{\phi,\beta}(u_\beta-l_\beta)\sum_{k=0}^{K-1}r_\beta(k)
+\frac{1}{2d}\sum_{\beta=1}^d \frac{(u_\beta-l_\beta)^2}{\eta_\beta}
\right).
\end{align*}

\subsection{The Nonconvex Case}
Now we consider the situation where $F$ is not assumed to be convex.

We start our analysis by observing that
\begin{align*}
\|\mathfrak{g}(x(k),\underline{L}_F)\|^2
=\ &
\left\|\underline{L}_F
\left(
x(k) - \mathcal{P}_{\mathcal{X}}\!\left[
x(k) - \frac{1}{\underline{L}_F}\nabla F(x(k))
\right]
\right)\right\|^2 \\
=\ &
\sum_{\alpha=1}^d
\left|\underline{L}_F
\left(
x_\alpha(k) - \mathcal{P}_{[l_\alpha,u_\alpha]}\!\left[
x_\alpha(k) - \frac{1}{\underline{L}_F}\frac{\partial F(x(k))}{\partial x_{\alpha}}
\right]
\right)\right|^2 \\
\leq\ &
\sum_{\alpha=1}^d
\left|\frac{1}{\eta_\alpha}
\left(
x_\alpha(k) - \mathcal{P}_{[l_\alpha,u_\alpha]}\!\left[
x_\alpha(k) - \eta_\alpha\frac{\partial F(x(k))}{\partial x_{\alpha}}
\right]
\right)\right|^2 \\
=\ &
d\,\mathbb{E}\!\left[
\left.\left|\frac{1}{\eta_{\alpha(k)}}
\left(
x_{\alpha(k)}(k) - \mathcal{P}_{[l_\alpha,u_\alpha]}\!\left[
x_{\alpha(k)}(k) - \eta_{\alpha(k)}\frac{\partial F(x(k))}{\partial x_{\alpha(k)}}
\right]
\right)\right|^2\right|\mathcal{F}_k
\right],
\end{align*}
where the third step follows from \cite[Lemma 2]{nesterov2013gradient} and the fact that $\underline{L}_F\leq L_{F,\alpha}\leq 1/\eta_\alpha$ for all $\alpha=1,\ldots,d$.
To bound the right-hand side of the above inequality, we note that
\begin{align*}
&
\frac{1}{\eta_{\alpha(k)}^2}\left|
x_{\alpha(k)}(k)-\mathcal{P}_{[l_{\alpha(k)},u_{\alpha(k)}]}\!\left[
x_{\alpha(k)}(k)-\eta_{\alpha(k)}\frac{\partial F(x(k))}{\partial x_{\alpha(k)}}
\right]
\right|^2 \\
\leq\ &
\frac{1+1/4}{\eta_{\alpha(k)}^2}|x_{\alpha(k)}(k)-x_{\alpha(k)}(k+1)|^2 \\
&
+\frac{1+4}{\eta_{\alpha(k)}^2}
\left|x_{\alpha(k)}(k+1)-\mathcal{P}_{[l_{\alpha(k)},u_{\alpha(k)}]}\!\left[x_{\alpha(k)}(k)-\eta_{\alpha(k)}\frac{\partial F(x(k))}{\partial x_{\alpha(k)}}\right]\right|^2 \\
\leq\ &
\frac{5}{4\eta_{\alpha(k)}^2}|x_{\alpha(k)}(k)-x_{\alpha(k)}(k+1)|^2
+5\left|g_{\alpha(k)}(k)-\frac{\partial F(x(k))}{\partial x_{\alpha(k)}}\right|^2 \\
\leq\ &
\frac{5}{4\eta_{\alpha(k)}^2}|x_{\alpha(k)}(k)-x_{\alpha(k)}(k+1)|^2
+\frac{5}{4}L_{\phi,\alpha(k)}^2 r(k)^2,
\end{align*}
where the second step follows from the nonexpansiveness of projection onto convex sets, and the last step follows from the bound~\eqref{eq:RZFCD_grad_diff}.
Then, we can use~\eqref{eq:proj_intermediate_1} to derive a bound on the first term:
\[
g_{\alpha(k)}(k)\,(x_{\alpha(k)}(k+1)-x_{\alpha(k)}(k))
\leq -\frac{1}{\eta_{\alpha(k)}}|x_{\alpha(k)}(k+1)-x_{\alpha(k)}(k)|^2,
\]
which, combined with \eqref{eq:RZFCD_smoothness_intermediate_1}, leads to
\begin{align*}
& F(x(k+1))
-F(x(k)) \\
\leq\ &
-\frac{1}{\eta_{\alpha(k)}}\left(
1
-\frac{\eta_{\alpha(k)}L_{F,\alpha(k)}}{2}\right)
|x_{\alpha(k)}(k+1)-x_{\alpha(k)}(k)|^2
+
\frac{1}{2}L_{\phi,\alpha(k)}(u_{\alpha(k)}-l_{\alpha(k)})r(k)
\\
\leq\ &
-\frac{1}{2\eta_{\alpha(k)}}
|x_{\alpha(k)}(k+1)-x_{\alpha(k)}(k)|^2
+
\frac{1}{2}L_{\phi,\alpha(k)}(u_{\alpha(k)}-l_{\alpha(k)})r(k)
.
\end{align*}
As a result,
\begin{align*}
& \frac{1}{\eta_{\alpha(k)}^2}\left|
x_{\alpha(k)}(k)-\mathcal{P}_{[l_{\alpha(k)},u_{\alpha(k)}]}\!\left[
x_{\alpha(k)}(k)-\eta_{\alpha(k)}\frac{\partial F(x(k))}{\partial x_{\alpha(k)}}
\right]
\right|^2 \\
\leq\ &
\frac{5}{2\eta_{\alpha(k)}}(F(x(k))-F(x(k+1)))
+\frac{5L_{\phi,\alpha(k)}r(k)}{4}
\left(
\frac{u_{\alpha(k)}-l_{\alpha(k)}}{\eta_{\alpha(k)}}
+L_{\phi,\alpha(k)} r(k)\right) \\
\leq\ &
\frac{5}{2\eta_{\alpha(k)}}(F(x(k))-F(x(k+1)))
+\frac{2L_{\phi,\alpha(k)}(u_{\alpha(k)}-l_{\alpha(k)})r(k)}{\eta_{\alpha(k)}},
\end{align*}
where in the last step we used $r(k)\leq (u_{\alpha(k)}-l_{\alpha(k)})/2$ and $\eta_{\alpha(k)}L_{\phi,\alpha(k)}\leq \eta_{\alpha(k)}L_{F,\alpha(k)}<6/5$. By taking the expectation conditioned on $\mathcal{F}_k$ and denoting $\underline{\eta}=\min_{\beta}\eta_\beta$, we get
\begin{align*}
\left\|\mathfrak{g}(x(k),\underline{L}_F)\right\|^2
\leq \ &
\frac{5d}{2\underline{\eta}}\,\mathbb{E}\!\left[F(x(k))-F(x(k+1))|\mathcal{F}_k\right]
+2 \sum_{\beta=1}^d
\frac{L_{\phi,\beta}(u_\beta-l_\beta)}{\eta_{\beta}}r_\beta(k)
\end{align*}
and by taking the total expectation and telescoping sum, we get
\begin{align*}
\frac{1}{K}\sum_{k=0}^{K-1}
\mathbb{E}\!\left[
\|\mathfrak{g}(x(k),\underline{L}_F)\|^2
\right]
\leq\ &
\frac{d}{K}
\left(
\frac{5}{2\underline{\eta}}
(F(0)-F^\ast)
+
\frac{2}{d}\sum_{\beta=1}^d
\frac{L_{\phi,\beta}(u_\beta-l_\beta)}{\eta_{\beta}}\sum_{k=0}^{K-1}r_\beta(k)
\right),
\end{align*}
which completes the proof.

%% file: case_study_new.tex
\section{Numerical Experiments}\label{section:case_studies}

In this section, we conduct numerical experiments to validate the performance of the proposed algorithms. Specifically, we first test our algorithms 2-ZFGD and RZFCD on a convex test case. Then, we conduct experiments on a nonconvex test case in which AC power flow and voltage constraints are taken into account.

\subsection{The Convex Test Case}
In the convex test case, the DDR problem consists of 100 agents coordinated by an aggregator as described in Section~\ref{section_3}. In this test case, only the load following requirement is considered in the global objective, i.e., the global objective function is given by $\phi(x) = \left(p_c(x)-D\right)^2$ without penalty associated with voltage safety. Furthermore, we employ an approximate model for the function $p_c$ given by
\[
p_c(x) = \sum\nolimits_{i=1}^{100} (1+\gamma_i)x_i
\]
(see~\cite{qin2018consensus}), with each loss-related coefficient $\gamma_i$ randomly selected from the interval $(0.03,0.15)$. Each feasible set $\mathcal{X}_i$ is set to $[0,u_i]$ where each $u_i$ is randomly selected from $(0\,\mathrm{kW},50\,\mathrm{kW})$. Each local cost function is a quadratic function $f_i(x) = a_ix_i^2 + b_i x_i$ with $a_i$ and $b_i$ randomly selected from $(0.5,1.5)$ and $(0,5)$, respectively. The desired load level is set to $D=\sum_i u_i-1500~\mathrm{kW}$.

We conduct experiments for both 2-ZFGD and RZFCD. For RZFCD, we test it with different \emph{constant step size (CS)} settings. For 2-ZFGD, apart from the constant step size settings that have been studied for theoretical analysis, we also test it under the \emph{diminishing step size (DS)} settings that are popular in stochastic optimization and distributed optimization.
For both algorithms, the smoothing radius $r(k)$ is set as $r(k)=\min\{1/(k+1)^{1.1},10^{-3}\}$. For 2-ZFGD under the setting of the diminishing step size, we set $\eta(k)=\eta(0)/\sqrt{k+1}$ with different $\eta (0)$; the parameter $\delta$ is also set to be diminishing as $\delta(k)=0.1/\sqrt{k+1}$ under the setting of both constant and diminishing step sizes.
Performances are evaluated using two metrics versus the iteration index $k$: i) the relative error~(RE) $(F(x(k))-F^\ast)/F^\ast$, where $F^\ast$ is the optimal value of the DDR problem; ii) the $\ell_2$ norm of the stationarity measure $\|\mathfrak{g}(x(k);M)\|$, where $M=1/0.3$ for all convex cases. All settings are tested for 50 random trials with the same algorithmic parameters, respectively.

\begin{figure}[tbp]
	\centering
	\includegraphics[width=.8\linewidth]{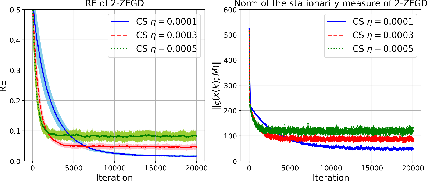}
	\caption{The performance of 2-ZFGD with constant step sizes in the convex case.}
	\label{fig1}
\end{figure}
\begin{figure}[tbp]
	\centering
	\includegraphics[width=.8\linewidth]{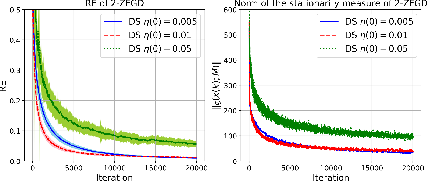}
	\caption{The performance of 2-ZFGD with diminishing step sizes in the convex case. }
	\label{fig2}
\end{figure}
\begin{figure}[tbp]
	\centering
	\includegraphics[width=.8\linewidth]{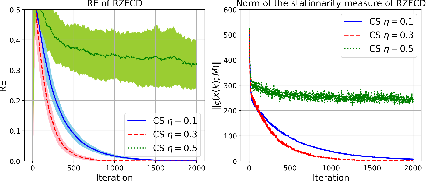}
	\caption{The performance of RZFCD with constant step sizes in the convex case.}
	\label{fig3}
\end{figure}

Figures~\ref{fig1}--\ref{fig3} illustrate the performance of 2-ZFGD and RZFCD. Here in the figures of the relative errors, dark curves represent the averaged values of 50 random trials, while the light shades represent the standard variance of all trials; in the figures of the norms of the stationarity measure, we only plot the average values for visual clarity. The figures show that both algorithms can converge with proper step sizes. Specifically, it can be seen that RZFCD can achieve a smaller final relative error than 2-ZFGD in both the CS and DS settings. For 2-ZFGD, we need to set the step size very small to achieve convergence due to the variance of gradient estimation, which also results in a much slower convergence rate. For RZFCD, the algorithm would fail to converge to the optimal solution if the step size is too large, which is in accordance with the condition $\eta_iL_i\leq 1$ in our theoretical analysis for RZFCD. However, if the step size is too small, the convergence of RZFCD will also become slow, which is typical behavior of first-order and zeroth-order methods.

To compare the two algorithms more clearly, we fix three levels of relative errors ($5\%$, $1\%$, and $0.1\%$), pick out the random trials with the best-tested parameters whose relative errors can drop below these levels within 20000 iterations, and compute the average numbers of iterations needed to achieve the three relative errors. We also compute the proportions out of $50$ runs that achieve these relative errors. The results are listed in Table~\ref{tab2}, in which ``N/A'' means no run can achieve such a relative error. It is obvious that 2-ZFGD needs much more iterations and zeroth-order queries for both settings. On the other hand, RZFCD achieves better performance with much fewer iterations. 

\begin{table}[H]
\centering
\setlength{\tabcolsep}{1mm}
{
\begin{tabular}{ccccc}
\toprule
\multicolumn{2}{c}{\multirow{2}{*}{Relative error}}      & 5\%                  & 1\%                  & 0.1\%                \\ \cmidrule{3-5} 
\multicolumn{2}{c}{}    & Iteration/Proportion & Iteration/Proportion & Iteration/Proportion \\ \midrule
\multirow{2}{*}{2-ZFGD} & CS $\eta=0.0001$ & 6234.5/$100\%$             & N/A                    & N/A                    \\ \cmidrule{2-5} 
                        & DS $\eta(0)=0.01$   & 2875.2/$100\%$             & 18435.4/$16\%$         & N/A                    \\ \midrule
RZFCD                   & CS $\eta=0.3$    & 376.7/$100\%$              & 621.4/$100\%$              & 981.7/$100\%$              \\
\bottomrule
\end{tabular}
}
\caption{Average numbers of iterations needed and proportions of trials for achieving certain levels of relative errors in  the convex case.}
\label{tab2}
\end{table}

\subsection{The Nonconvex Test Case}

In the nonconvex test case, we consider an aggregator coordinating multiple agents in the distribution feeder. The distribution feeder is based on the 141-bus radial system from~\cite{khodr2008maximum}, of which we adopt the topology and the line parameters. Decision variables include active and inactive power loads at all buses, whose lower bounds $l_i$ are zero and upper bounds $u_i$ are the nominated load levels from the original 141-bus system. The nonlinear and nonconvex AC power flows are incorporated to formulate penalties for voltage safety. The global objective functions are then given by~\eqref{eq:global_cost}, including the squared difference to the desired load level $(p_c(x)-D)^2$ as well as the penalty term $\rho(x)$ for voltage safety. We set $\underline{v}=0.96~\mathrm{p.u.}$, $\overline{v}=1.04~\mathrm{p.u.}$ and $\alpha_D=\alpha_v=20$ in our test case. The parameters $a_i,\ b_i$ are generated in the same way as in the convex test case. We assume that the aggregator can observe or measure the total active power fed into the network, as well as the voltages of all buses. The total load to be curtailed is set to $0.15~\mathrm{p.u.}=1500~ \mathrm{kW}$ in this case.

We test the RZFCD algorithm with constant step size as well as the 2-ZFGD algorithm with both settings of constant and diminishing step sizes. For RZFCD, the constant step size is set to be $\eta=0.025$ and the smoothing radius is set to be $r(k)=\min\{0.1/(k+1)^{1.2},2\times10^{-4}\}$. For 2-ZFGD under the constant step size setting, we set $\eta=3\times 10^{-6}$, $\delta=0.005$ and $r(k)=\min\{0.01/(k+4000)^{1.1},10^{-5}\}$. For 2-ZFGD under the diminishing step size setting, we set $\eta(k)=3\times 10^{-4}/\sqrt{k+1000}$, $\delta(k) = \min\{50/(k+1),0.1\}$ and $r(k)$ the same as 2-ZFGD under the constant step size setting. These parameters are tuned in such a way that the empirical convergence of the algorithms can be as fast as possible.

\begin{figure}[tbp]
\centering
\begin{subfigure}{.325\textwidth}
\centering
\includegraphics[width=\linewidth]{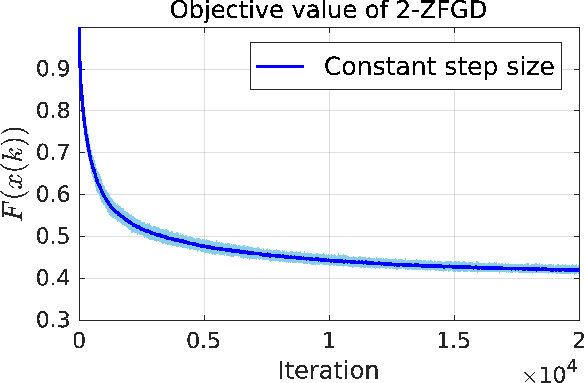} \\
\vspace{6pt}
\includegraphics[width=\linewidth]{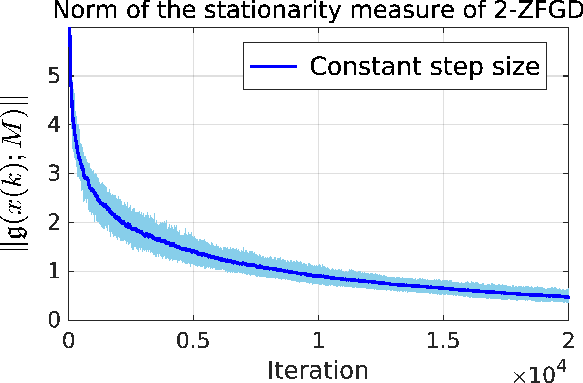}
\caption{2-ZFGD, constant stepsize}
\end{subfigure}
\begin{subfigure}{.325\textwidth}
\centering
\includegraphics[width=\linewidth]{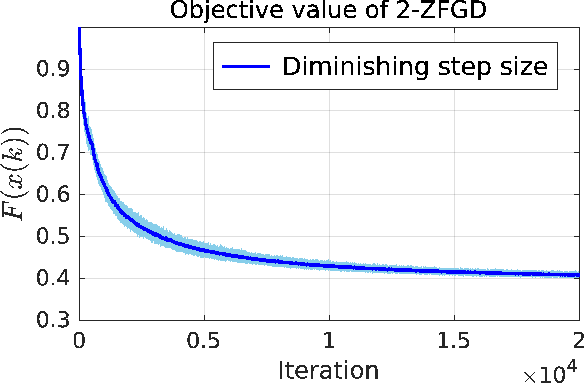} \\
\vspace{6pt}
\includegraphics[width=\linewidth]{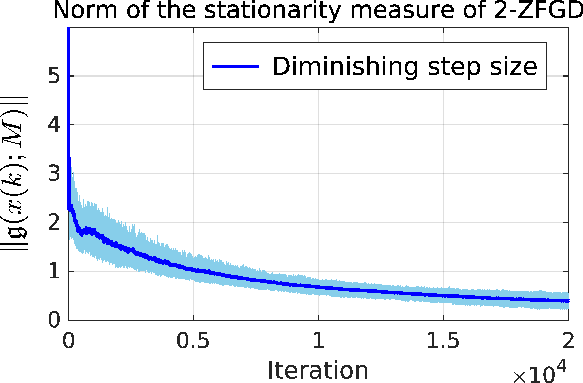}
\caption{2-ZFGD, diminishing stepsize}
\end{subfigure}
\begin{subfigure}{.325\textwidth}
\centering
\includegraphics[width=\linewidth]{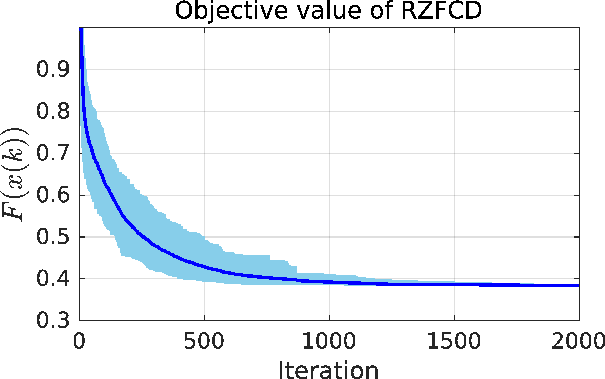} \\
\vspace{6pt}
\includegraphics[width=\linewidth]{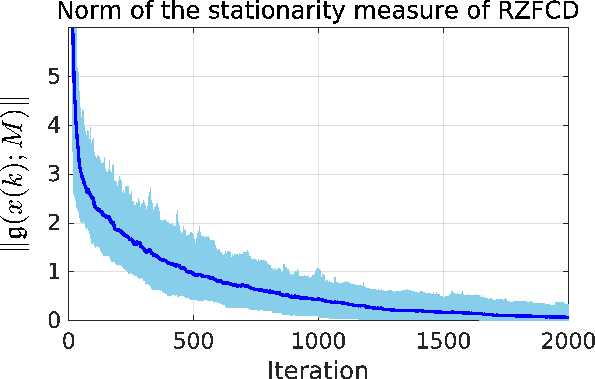}
\caption{RZFCD}
\end{subfigure}
\caption{The performance of 2-ZFGD and RZFCD in the nonconvex test case.}
\label{fig:nonconvex_numerical}
\end{figure}

Figure~\ref{fig:nonconvex_numerical} illustrates the numerical convergence behavior of the three settings, where we plot the objective value $F(x(k))$ and the norm of the stationarity measure $\|\mathfrak{g}(x(k);M)\|$ with $M=1/0.025$ versus $k$. Here, each dark solid curve shows the average trajectory of 100 random trials for each setting, and the light blue shades represent the interval from the 5th percentile to the 95th percentile among the 100 random trials. It can be seen that RZFCD with a proper constant step size converges much faster than 2-ZFGD under both the constant and the diminishing step size settings. These observations justify our theoretical results, and also suggest that 2-ZFGD does not seem to be able to compete with RZFCD even if we employ diminishing step sizes.

%% file: appendix_new.tex
\section{Auxiliary Results for 2-ZFGD}

Let $S(x,r)$ denote the Cartesian product $\prod_{i=1}^N S_i(x_i,r)$.

\begin{lemma}
\label{2_point_estimator_error}
Suppose $\mathcal{X}\subset\mathbb{R}^d$ is a compact and convex set satisfying $\underline{R}\mathbb{B}_d\subseteq\mathcal{X}\subseteq\overline{R}\mathbb{B}_d$. Let $h:\mathcal{X}\rightarrow\mathbb{R}$ be $\Lambda$-Lipschitz and $L$-smooth. Let $\delta\in(0,1)$ and $r\in(0,\delta\underline{R}/(3\sqrt{d}))$. Then
\[
\left\|\mathbb{E}_z\!\left[
G_h(x;r,z)
\right]
-\nabla h^r(x)
\right\|
\leq \Lambda\left(\frac{4\overline{R}^2}{r^2}+2\right)
\exp\!\left(\frac{d}{2}-\frac{\delta^2\underline{R}^2}{4r^2}\right)
\]
for all $x\in(1-\delta)\mathcal{X}$, where $z\sim \mathcal{Z}(x,r)$, and $h^r:(1-\delta)\mathcal{X}\rightarrow\mathbb{R}$ is a $\Lambda$-Lipschitz and $L$-smooth function satisfying
\[
|h^r(x)-h(x)|\leq \min\left\{r\Lambda\sqrt{d},\frac{1}{2}r^2Ld\right\},
\quad
\|\nabla h^r(x)- \nabla h(x)\|
\leq rL\sqrt{d}.
\]
\end{lemma}
\begin{proof}
This lemma is an extension of \cite[Lemma 1]{tang2023zeroth}. Denote $\beta=\delta\underline{R}/(r\sqrt{d})$. By following the proof of \cite[Lemma 1]{tang2023zeroth}, we can show that
\[
\left\|\mathbb{E}_{z\sim\mathcal{Z}(x,r)}\!\left[
G_h(x;r,z)
\right]
-\kappa(r)\nabla h^r(x)
\right\|
\leq \frac{4\Lambda\overline{R}^2}{r^2}
\exp\!\left(\frac{d}{2}-\frac{\delta^2\underline{R}^2}{4r^2}\right).
\]
Here $h^r:(1-\delta)\mathcal{X}\rightarrow\mathbb{R}$ is defined by
\[
h^r(x)=\mathbb{E}_{y\sim\mathcal{Y}(r)}[h(x+ry)]
\]
with $\mathcal{Y}(r)$ being an isotropic probability distribution satisfying $\mathbb{E}_{y\sim\mathcal{Y}(r)}[\|y\|^2]\leq d$;
the quantity $\kappa(r)$ satisfies
\[
1-\kappa(r)
\leq \frac{1}{\sqrt{\pi d}}
\left(\beta e^{-(\beta^2-1)/2}\right)^{\!d}
+\left(\beta^2 e^{1-\beta^2}\right)^{\!d/2},
\]
where $\beta=r^{-1}\delta \underline{R}/\sqrt{d}$. By using $x\leq e^{x^2/4}$ and $x^2\leq e^{x^2/2}$ for $x>0$, we can derive that
\[
1-\kappa(r)
\leq \left(\frac{1}{\sqrt{\pi d}}
+1\right)\exp\!\left(\frac{d}{2}-\frac{\beta^2 d}{4}\right)
\leq 2 \exp\!\left(\frac{d}{2}-\frac{\beta^2 d}{4}\right).
\]
As a result,
\begin{align*}
& \left\|\mathbb{E}_{z\sim\mathcal{Z}(x,r)}\!\left[
G_h(x;r,z)
\right]
-\nabla h^r(x)
\right\| \\
\leq\ &
\left\|\mathbb{E}_{z\sim\mathcal{Z}(x,r)}\!\left[
G_h(x;r,z)
\right]
-\kappa(r)\nabla h^r(x)
\right\|
+(1-\kappa(r))\|\nabla h^r(x)\| \\
\leq\ &
\frac{4\Lambda\overline{R}^2}{r^2}
\exp\!\left(\frac{d}{2}-\frac{\beta^2 d}{4}\right)
+\Lambda(1-\kappa(r))
\leq \Lambda\left(\frac{4\overline{R}^2}{r^2}+2\right)
\exp\!\left(\frac{d}{2}-\frac{\beta^2 d}{4}\right).
\end{align*}
The bound for $|h^r(x)-h(x)|$ can be proved similarly as in the proof of \cite[Lemma 1]{tang2023zeroth}; the bound for $\|\nabla h^r(x)-\nabla h(x)\|$ can be derived by
\begin{align*}
\|\nabla h^r(x)-\nabla h(x)\|
=\ &
\|\mathbb{E}_{y\sim\mathcal{Y}(r)}[\nabla_x h(x+ry)-\nabla_x h(x)]\| \\
\leq\ &
\mathbb{E}_{y\sim\mathcal{Y}(r)}\!\left[
\|\nabla_x h(x+ry)-\nabla_x h(x)\|\right] \\
\leq\ &
\mathbb{E}_{y\sim\mathcal{Y}(r)}
\!\left[L\cdot r\|y\|\right]
\leq rL\sqrt{d}.
\qedhere
\end{align*}
\end{proof}

\begin{lemma}
\label{lemma:2point_second_moment}
Let $h:\mathcal{X}\rightarrow\mathbb{R}$ be $L$-smooth. Then for any $\delta\in(0,1)$,
\[
\mathbb{E}_z\!\left[
\left\|G_h(x;r,z)\right\|^2
\right]
\leq
2\left[d+2+\left(\frac{2\overline{R}}{r}\right)^{\!4}
\exp\left(\frac{d}{2}-\frac{\delta^2\underline{R}^2}{4r^2}\right)\right]\|\nabla h(x)\|^2
+ \frac{r^2L^2(d+6)^3}{2},
\]
where $z\sim\mathcal{Z}(x,r)$.
\end{lemma}
\begin{proof}
Let $x\in\mathcal{X}$ and $r>0$ be fixed.
For notational simplicity, we denote $\mathcal{P}z= \mathcal{P}_{S(x,r)}[z]$ and $\bar{h}(z)=h(x+rz)$ for any $z\in\mathbb{R}^d$. It's not hard to check that $\bar{h}$ is $r^2L$-smooth, and so
\begin{align*}
\left|\bar{h}(\mathcal{P}z)-\bar{h}(0)-\langle\nabla \bar{h}(0),\mathcal{P}z\rangle
\right|
\leq \frac{r^2L}{2}\|\mathcal{P}z\|^2
\leq \frac{r^2L}{2}\|z\|^2,
\end{align*}
which leads to
\begin{equation}\label{eq:variance_bound_temp1}
\begin{aligned}
& \mathbb{E}_{z\sim\mathcal{Z}(x,r)}\!\left[
\left\|G_h(x;r,z)-\langle \nabla h(x),z\rangle z\right\|^2\right] \\
=\ &
\frac{1}{r^2}\,\mathbb{E}_{z\sim\mathcal{N}(0,I_d)}\!\left[
\left|\bar{h}(\mathcal{P}z)-\bar{h}(0)-\langle\nabla \bar{h}(0),\mathcal{P}z\rangle
\right|^2
\|\mathcal{P}z\|^2\right] \\
\leq\ &
\frac{1}{r^2}\cdot\frac{r^4L^2}{4}\,
\mathbb{E}_{z\sim\mathcal{N}(0,I_d)}\!\left[
\|z\|^6\right]
\leq \frac{r^2L^2}{4}(d+6)^3.
\end{aligned}
\end{equation}
Then, since the distribution $\mathcal{Z}(x,r)$ has a density function in the interior of $\frac{\delta\underline{R}}{r}\mathbb{B}_d$ that coincides with the standard Gaussian distribution $p_{\mathcal{N}(0,I_d)}(z)$, we have
\begin{align*}
& \mathbb{E}_{z\sim\mathcal{Z}(x,r)}
\!\left[
|\langle\nabla h(x),z\rangle|^2 \|z\|^2\right]
-\int_{\mathbb{R}^d}
|\langle\nabla h(x),z\rangle|^2 \|z\|^2
p_{\mathcal{N}(0,I_d)}(z)
\cdot\mathsf{1}_{\|z\|<\delta\underline{R}/r}(z)\,dz \\
=\ &
\mathbb{E}_{z\sim\mathcal{Z}(x,r)}
\!\left[
\left|\langle\nabla h(x),z\rangle\right|^2 \|z\|^2
\mathsf{1}_{\|z\|\geq \delta\underline{R}/r}(z)
\right] \\
\leq\ &
\sup_{z\in S(x,r),\|z\|\geq \delta\underline{R}/r} \left(
\left|\langle\nabla h(x),z\rangle\right|^2 \|z\|^2
\right)
\cdot
\left(1-\mathbb{P}_{z\sim\mathcal{Z}(x,r)}(\|z\|<\delta\underline{R}/r)\right) \\
\leq\ &
\|\nabla h(x)\|^2
\sup_{z\in S(x,r),\|z\|\geq \delta\underline{R}/r}\|z\|^4
\cdot\mathbb{P}_{z\sim\mathcal{N}(0,I_d)}\!\left(\sum_{i=1}^d z_i^2\geq\frac{\delta^2\underline{R}^2}{r^2}\right) \\
\leq\ &
\|\nabla h(x)\|^2 \left(\frac{2\overline{R}}{r}\right)^4
\exp\left(\frac{d}{2}-\frac{\delta^2\underline{R}^2}{4r^2}\right),
\end{align*}
where we used the inequality
\[
\mathbb{P}_{z\sim\mathcal{N}(0,I_d)}\left(\sum_{i=1}^d z_i^2\geq \beta^2 d\right)\leq \exp\!\left(\frac{d}{2}-\frac{\beta^2d}{4}\right).
\]
Next, we notice that
\begin{align*}
& \int_{\mathbb{R}^d}
|\langle\nabla h(x),z\rangle|^2 \|z\|^2
p_{\mathcal{N}(0,I_d)}(z)
\cdot\mathsf{1}_{\|z\|<\delta\underline{R}/r}(z)\,dz \\
\leq\ &
\int_{\mathbb{R}^d}
|\langle\nabla h(x),z\rangle|^2 \|z\|^2
p_{\mathcal{N}(0,I_d)}(z)\,dz \\
=\ &
\mathbb{E}_{z\sim\mathcal{N}(0,I_d)}
\!\left[|\langle \nabla h(x),z\rangle|^2 \|z\|^2 \right] \\
=\ &
\nabla h(x)^\tran\cdot \mathbb{E}_{z\sim\mathcal{N}(0,I_d)}\!\left[\|z\|^2zz^\tran\right]
\cdot \nabla h(x) \\
=\ &
\nabla h(x)^\tran \cdot (d+2)I_d \cdot \nabla h(x)
= (d+2)\|\nabla h(x)\|^2.
\end{align*}
As a result, we have
\[
\mathbb{E}_{z\sim\mathcal{Z}(x,r)}
\!\left[
|\langle\nabla h(x),z\rangle|^2 \|z\|^2\right]
\leq 
(d+2)\|\nabla h(x)\|^2
+
\Lambda^2 \left(\frac{2\overline{R}}{r}\right)^4
\exp\left(\frac{d}{2}-\frac{\delta^2\underline{R}^2}{4r^2}\right).
\]
We can now upper bound $\mathbb{E}_{z\sim\mathcal{Z}(x,r)}\!\left[
\left\|G_h(x;r,z)\right\|^2
\right]$ by
\begin{align*}
& \mathbb{E}_{z\sim\mathcal{Z}(x,r)}\!\left[
\left\|G_h(x;r,z)\right\|^2
\right] \\
\leq\ &
2\,\mathbb{E}_{z\sim\mathcal{Z}(x,r)}\!\left[
\left\|G_h(x;r,z)-\langle \nabla h(x),z\rangle z \right\|^2
\right]
+2\,\mathbb{E}_{z\sim\mathcal{Z}(x,r)}\!\left[
\|\langle\nabla h(x),z\rangle z \|^2
\right] \\
\leq\ &
2\left[d+2+\left(\frac{2\overline{R}}{r}\right)^{\!4}
\exp\left(\frac{d}{2}-\frac{\delta^2\underline{R}^2}{4r^2}\right)\right]\|\nabla h(x)\|^2
+ \frac{r^2L^2(d+6)^3}{2}.
\qedhere
\end{align*}
\end{proof}

\begin{lemma}\label{lemma:shrink_proj_difference}
Let $\mathcal{X}\subset\mathbb{R}^d$ be a convex and compact set, and let $\delta\in(0,1)$ be arbitrary. Then for any $x\in\mathbb{R}^d$, we have
\[
\|\mathcal{P}_{\mathcal{X}}[x]-\mathcal{P}_{(1-\delta)\mathcal{X}}[x]\|
\leq \delta \overline{R}.
\]
\end{lemma}
\begin{proof}
Since $(1-\delta)\mathcal{P}_{\mathcal{X}}[x]\in(1-\delta)\mathcal{X}$ and $(1-\delta)^{-1}\mathcal{P}_{(1-\delta)\mathcal{X}}[x]\in\mathcal{X}$, by the properties of projection operators onto convex sets, we have
\begin{align*}
\left\langle x-\mathcal{P}_{\mathcal{X}}[x],
\frac{1}{1-\delta}\mathcal{P}_{(1-\delta)\mathcal{X}}[x]-\mathcal{P}_{\mathcal{X}}[x]\right\rangle\ &
\leq 0, \\
\left\langle x-\mathcal{P}_{(1-\delta)\mathcal{X}}[x],
(1-\delta)\mathcal{P}_{\mathcal{X}}[x]
-\mathcal{P}_{(1-\delta)\mathcal{X}}[x]\right\rangle\ &
\leq 0.
\end{align*}
By multiplying the first inequality with $1-\delta$ and adding it to the second inequality, we get
\[
\left\langle
\mathcal{P}_{(1-\delta)\mathcal{X}}[x]
-\mathcal{P}_{\mathcal{X}}[x],
\mathcal{P}_{(1-\delta)\mathcal{X}}[x]
-(1-\delta)\mathcal{P}_{\mathcal{X}}[x]
\right\rangle\leq 0.
\]
We then see that
\begin{align*}
\|\mathcal{P}_{(1-\delta)\mathcal{X}}[x]
-\mathcal{P}_{\mathcal{X}}[x]\|^2
\leq\ &
-\delta \left\langle
\mathcal{P}_{(1-\delta)\mathcal{X}}[x]
-\mathcal{P}_{\mathcal{X}}[x],
\mathcal{P}_{\mathcal{X}}[x]
\right\rangle \\
\leq\ &
\delta \|\mathcal{P}_{(1-\delta)\mathcal{X}}[x]
-\mathcal{P}_{\mathcal{X}}[x]\|
\|\mathcal{P}_{\mathcal{X}}[x]\| \\
\leq\ &
\delta \overline{R} \|\mathcal{P}_{(1-\delta)\mathcal{X}}[x]
-\mathcal{P}_{\mathcal{X}}[x]\|,
\end{align*}
from which we can directly obtain the desired bound.
\end{proof}

\section{Proof of Theorem~\ref{theorem:2ZFGD_convergence}}
\label{proof_of_theorem1}

We denote the filtration $\mathcal{F}_k = \sigma(x(1),z(1),\ldots,x(k-1),z(k-1),x(k))$. We also denote
\[
x(k) = (x_1(k),\ldots,x_N(k)),
\quad
g(k) = (g_1(k),\ldots,g_N(k)),
\]
so that the iterations of 2-ZFGD can be equivalently written as
\[
x(k+1) = \mathcal{P}_{(1-\delta)\mathcal{X}}
\!\left[x(k) - \eta\, g(k)\right].
\]

First, we introduce the following general result for the analysis of projected-SGD-type algorithms.
\begin{lemma}
Consider the iterations $x(k+1)=\mathcal{P}_{(1-\delta)\mathcal{X}}\!\left[
x(k)-\eta\,g(k)
\right]$. Suppose $F$ is $L_F$-smooth and $\eta\leq 1/(2L_F)$. Then we have
\begin{equation}
\label{mirror_descent_result}
\begin{aligned}
\frac{1}{2\eta}(\left\|\tilde{x}-x(k+1)\right\|^2-\left\|\tilde{x}-x(k)\right\|^2)
\leq\ &
\langle g(k),\tilde{x}-x(k)\rangle+
\eta \|g(k)-\nabla F(x(k))\|^2 \\
& + F(x(k))-F(x(k+1))
\end{aligned}
\end{equation}
\end{lemma}
\begin{proof}
Let $\tilde{x}\in(1-\delta)\mathcal{X}$ be arbitrary. Note that $x(k+1)=\mathcal{P}_{(1-\delta)\mathcal{X}}\!\left[
x(k)-\eta\,g(k)
\right]$ implies
\[
\langle
x(k)-\eta\,g(k)-x(k+1), \tilde{x}-x(k+1)\rangle\leq 0.
\]
By using $\|x(k)-x(k+1)\|^2+\|\tilde{x}-x(k+1)\|^2-\|\tilde{x}-x(k)\|^2=2\langle x(k)-x(k+1),\tilde{x}-x(k+1)\rangle$, we can derive from the above inequality that
\begin{align*}
\frac{1}{2\eta}
\left(\|\tilde{x}-x(k+1)\|^2-\|\tilde{x}-x(k)\|^2\right)
\leq\ &
\langle g(k),\tilde{x}-x(k+1)\rangle
-\frac{1}{2\eta}\|x(k+1)-x(k)\|^2.
\end{align*}
The inner product term can be further bounded by
\begin{align*}
\langle g(k),\tilde{x}-x(k+1)\rangle
=\ & 
\langle g(k),\tilde{x}-x(k)\rangle
+ \langle g(k)-\nabla F(x(k)),x(k)-x(k+1)\rangle \\
&
+ \langle \nabla F(x(k)),x(k)-x(k+1)\rangle \\
\leq\ &
\langle g(k),\tilde{x}-x(k)\rangle
+ \eta\|g(k)-\nabla F(x(k))\|^2
+\frac{1}{4\eta}\|x(k)-x(k+1)\|^2 \\
&
+ \langle \nabla F(x(k)), x(k)-x(k+1)\rangle,
\end{align*}
and by plugging in $F(x(k+1))\leq F(x(k)) + \langle\nabla F(x(k)),x(k+1)-x(k)\rangle
+\frac{L_F}{2}\|x(k+1)-x(k)\|^2$ which is a consequence of the $L_F$-smoothness of $F$, we get
\begin{align*}
\langle g(k),\tilde{x}-x(k+1)\rangle
\leq\ &
\langle g(k),\tilde{x}-x(k)\rangle
+\eta\|g(k)-\nabla F(x(k))\|^2 \\
&
+\left(\frac{L_F}{2}+\frac{1}{4\eta}\right)\|x(k)-x(k+1)\|^2
+F(x(k))-F(x(k+1)).
\end{align*}
Combining all previous results, we get
\begin{align*}
& \frac{1}{2\eta}
\left(\|\tilde{x}-x(k+1)\|^2-\|\tilde{x}-x(k)\|^2\right) \\
\leq\ &
\langle g(k),\tilde{x}-x(k)\rangle
+ \eta \|g(k)-\nabla F(x(k))\|^2 \\
&
+ F(x(k))-F(x(k+1))
+\left(\frac{L_F}{2}-\frac{1}{4\eta}\right)
\|x(k)-x(k+1)\|^2.
\end{align*}
We complete the proof by using the condition that $\eta\leq 1/(2L_F)$.
\end{proof}

Our analysis consists of the following steps:

\noindent\textbf{1. Bound the expectation of the right-hand side of (\ref{mirror_descent_result})}. The first term on the right-hand side of (\ref{mirror_descent_result}) can be bounded by the following lemma.
\begin{lemma}
\label{term1_lemma}
Assume $0<r(k)\leq\frac{\delta \underline{R}}{3\sqrt{d}}$ and $\delta\in(0,1)$. Then for any $\tilde{x}\in (1-\delta)\mathcal{X}$,
\begin{align*}
\mathbb{E}\left[\langle g(k),\tilde{x}-x(k)\rangle\right]
\leq\ &\mathbb{E}\left[F(\tilde{x})-F(x(k))\right]
+
2r(k)L_\phi \overline{R}\sqrt{d}
+
2\Lambda_\phi \overline{R}\left(\frac{4\overline{R}^2}{r(k)^2}+2\right)\exp{\left(\frac{d}{2}-\frac{\delta^2\underline{R}^2}{4r(k)^2}\right)}.
\end{align*}
\end{lemma}
\begin{proof}
We have
\begin{align*}
\mathbb{E}\!\left[\langle g(k),\tilde{x}-x(k)\rangle|\mathcal{F}_k\right]
=\ &
\langle
\mathbb{E}\!\left[g(k)|\mathcal{F}_k\right],\tilde{x}-x(k)\rangle \\
=\ &
\langle \nabla F(x(k))+\nabla \phi^{r(k)}(x(k))-\nabla \phi(x(k)),\tilde{x}-x(k)\rangle \\
&
+
\left\langle
\mathbb{E}\!\left[g(k)|\mathcal{F}_k\right]
-(\nabla F(x(k))-\nabla\phi(x(k)))
-\nabla \phi^{r(k)}(x(k)),\tilde{x}-x(k)\right\rangle.
\end{align*}
For the first term on the right-hand side, we have
\begin{align*}
& \langle \nabla F(x(k))+\nabla \phi^{r(k)}(x(k))-\nabla \phi(x(k)),\tilde{x}-x(k)\rangle \\
\leq\ &
\langle \nabla F(x(k)), \tilde{x}-x(k)\rangle
+
\|\nabla\phi^{r(k)}(x(k))-\nabla\phi(x(k))\|
\cdot\|\tilde{x}-x(k)\| \\
\leq\ &
F(\tilde{x})-F(x(k))
+ r(k)L_\phi\sqrt{d}\cdot 2\overline{R},
\end{align*}
where we used the convexity of $F$ and Lemma~\ref{2_point_estimator_error}. For the second term, we have
\begin{align*}
& \langle
\mathbb{E}\!\left[g(k)|\mathcal{F}_k\right]
-(\nabla F(x(k))-\nabla\phi(x(k)))
-\nabla \phi^{r(k)}(x(k)),\tilde{x}-x(k)\rangle \\
\leq\ &
\|\mathbb{E}\!\left[g(k)|\mathcal{F}_k\right]
-(\nabla F(x(k))-\nabla\phi(x(k)))
-\nabla \phi^{r(k)}(x(k))\|
\cdot \|\tilde{x}-x(k)\| \\
\leq\ &
\Lambda_\phi\left(\frac{4\overline{R}^2}{r(k)^2}+2\right)\exp{\left(\frac{d}{2}-\frac{\delta^2\underline{R}^2}{4r(k)^2}\right)}\cdot 2\overline{R},
\end{align*}
where we used Lemma~\ref{2_point_estimator_error} in the last step. Summarizing these bounds and taking the total expectation, we get the desired inequality.
\end{proof}

For the expectation of the second term on the right-hand side of (\ref{mirror_descent_result}), we note that
\begin{align*}
\mathbb{E}\!\left[\|g(k)-\nabla F(x(k))\|^2
|\mathcal{F}_k\right]
=\ &
\mathbb{E}\!\left[
\|G_\phi(x(k);r(k),z(k))
-\nabla\phi(x(k))\|^2
|\mathcal{F}_k
\right] \\
=\ &
\mathbb{E}\!\left[
\|G_\phi(x(k);r(k),z(k))\|^2
|\mathcal{F}_k\right]
- \|\nabla\phi(x(k))\|^2 \\
&
-
2\left\langle\mathbb{E}\!\left[G_\phi(x(k);r(k),z(k))|\mathcal{F}_k\right]
-\nabla\phi(x(k)),\nabla \phi(x(k))\right\rangle \\
\leq\ &
\mathbb{E}\!\left[
\|G_\phi(x(k);r(k),z(k))\|^2
|\mathcal{F}_k\right]-\|\nabla\phi(x(k))\|^2 \\
&
+ 2\left\|
\mathbb{E}\!\left[
G_\phi(x(k);r(k),z(k))|\mathcal{F}_k\right]-\nabla\phi(x(k))\right\|\cdot\|\nabla\phi(x(k))\|,
\end{align*}
By Lemma~\ref{2_point_estimator_error} and the assumption that $\phi$ is $\Lambda_\phi$-Lipschitz and $L_\phi$-smooth, we have
\begin{align*}
& 2\left\|\mathbb{E}\!\left[
G_\phi(x(k);r(k),z(k))|\mathcal{F}_k\right]-\nabla\phi(x(k))\right\|\cdot \|\nabla\phi(x(k))\| \\
\leq\ &
2\left(\left\|
\mathbb{E}\!\left[
G_\phi(x(k);r(k),z(k))|\mathcal{F}_k\right]-\nabla\phi^{r(k)}(x(k))\right\|
+\left\|
\nabla\phi^{r(k)}(x(k))-\nabla\phi(x(k))\right\|
\right)\|\nabla\phi(x(k))\|\\
\leq\ &
2\Lambda_\phi\left\|
\mathbb{E}\!\left[
G_\phi(x(k);r(k),z(k))|\mathcal{F}_k\right]-\nabla\phi^{r(k)}(x(k))\right\|
+
\left\|
\nabla\phi^{r(k)}(x(k))-\nabla\phi(x(k))\right\|^2
+ \|\nabla\phi(x(k))\|^2 \\
\leq\ &
2\Lambda_\phi^2\left(\frac{4\overline{R}^2}{r(k)^2}+2\right)
\exp\!\left(\frac{d}{2}-\frac{\delta^2\overline{R}^2}{4r(k)^2}\right)
+r(k)^2 L_\phi^2 d + \|\nabla\phi(x(k))\|^2,
\end{align*}
and by Lemma~\ref{lemma:2point_second_moment},
\begin{align*}
\mathbb{E}\!\left[
\|G_\phi(x(k);r(k),z(k))\|^2
|\mathcal{F}_k\right]
\leq\ &
2\Lambda_\phi^2\left[d+2+\left(\frac{2\overline{R}}{r(k)}\right)^{\!4}
\exp\left(\frac{d}{2}-\frac{\delta^2\underline{R}^2}{4r(k)^2}\right)\right]
+ \frac{r(k)^2L^2(d+6)^3}{2}.
\end{align*}
By summarizing the previous bounds, we obtain
\begin{equation}\label{eq:var_gk_bound}
\begin{aligned}
\mathbb{E}\!\left[
\|g(k)-\nabla F(x(k))\|^2
|\mathcal{F}_k\right]
\leq\ &
2\Lambda_\phi^2(d+2)
+ r(k)^2L_\phi^2(d+5)^3\\
&
+
2\Lambda_\phi^2
\left[\left(\frac{2\overline{R}}{r(k)}\right)^{\!4}
+\left(\frac{2\overline{R}}{r(k)}\right)^{\!2}
+2
\right]
\exp\left(\frac{d}{2}-\frac{\delta^2\underline{R}^2}{4r(k)^2}\right),
\end{aligned}
\end{equation}
where we used $(d+6)^3/2+d\leq (d+5)^3$. We remark that the bound ~\eqref{eq:var_gk_bound} applies also to the nonconvex case.

We can now take the expectation of~\eqref{mirror_descent_result} and plug in all the derived bounds to get
\begin{equation}
\label{eq:2ZFGD_convex_step1_final}
\begin{aligned}
&\frac{1}{2\eta}\left(\|\tilde{x}-x(k+1)\|^2
-\|\tilde{x}-x(k)\|^2\right) \\
\leq\ &
\mathbb{E}\left[F(\tilde{x})-F(x(k+1))\right]
+
2r(k)L_\phi \overline{R}\sqrt{d}
+
2\eta\Lambda_\phi^2(d+2)
+ \eta r(k)^2L_\phi^2(d+5)^3 \\
&
+2\Lambda_\phi
\left[\overline{R}\left(\frac{4\overline{R}^2}{r(k)^2}+2\right)
+
\eta\Lambda_\phi
\left(\left(\frac{2\overline{R}}{r(k)}\right)^{\!4}
+\left(\frac{2\overline{R}}{r(k)}\right)^{\!2}
+2
\right)
\right]\exp{\left(\frac{d}{2}-\frac{\delta^2\underline{R}^2}{4r(k)^2}\right)}.
\end{aligned}
\end{equation}

\noindent\textbf{2. Take the telescoping sum}. Let $x^\ast$ denote the optimizer of $F(x)$ over $x\in\mathcal{X}$, and let $\tilde{x}=(1-\delta)x^\ast$. Since $\|x^\ast-\tilde{x}\|=\delta\|x^\ast\|\leq\delta\overline{R}$ and $F$ is $\Lambda_F$-Lipschitz, we have
$F(\tilde{x})\leq F(x^\ast) + \delta\overline{R}\Lambda_F$. By combining it with~\eqref{eq:2ZFGD_convex_step1_final} and taking the telescoping sum, we can derive
\begin{equation}
\label{eq:2ZFGD_convex_step2_beginning}
\begin{aligned}
\mathbb{E}\!\left[
\frac{1}{K}\sum_{k=1}^K
F(x(k))
\right]-F(x^\ast)
\leq\ &
\delta\overline{R}\Lambda_F
+\frac{\|\tilde{x}-x(0)\|^2}{2\eta K}
+
2\eta\Lambda_\phi^2(d+2)
+\frac{2L_\phi\overline{R}\sqrt{d}}{K}\sum_{k=0}^{K-1}r(k) \\
&
+\frac{\eta L_\phi^2(d+5)^3}{K}\sum_{k=0}^{K-1}
r(k)^2
+\frac{1}{K}\sum_{k=0}^{K-1}
\omega_k,
\end{aligned}
\end{equation}
where we denote
\begin{align*}
\omega_k
=
2\Lambda_\phi
\left[\overline{R}\left(\frac{4\overline{R}^2}{r(k)^2}+2\right)
+
\eta\Lambda_\phi
\left(\left(\frac{2\overline{R}}{r(k)}\right)^{\!4}
+\left(\frac{2\overline{R}}{r(k)}\right)^{\!2}
+2
\right)
\right]\exp{\left(\frac{d}{2}-\frac{\delta^2\underline{R}^2}{4r(k)^2}\right)}.
\end{align*}
Next, we use the conditions on the algorithmic parameters to show that the right-hand side of~\eqref{eq:2ZFGD_convex_step2_beginning} is upper bounded by $\epsilon$. Indeed, by the condition on $\delta$ we have $\delta\overline{R}\Lambda_F\leq \epsilon/5$. Then by the conditions on $\eta$ and $K$, we get
\[
2\eta\Lambda_\phi^2(d+2)\leq \frac{\epsilon}{5},
\qquad
\frac{\|\tilde{x}-x(0)\|^2}{2\eta K}
\leq\frac{\epsilon}{5}.
\]
Moreover,
\[
\frac{2L_\phi\overline{R}\sqrt{d}}{K}\sum_{k=0}^{K-1}r(k)
\leq \frac{1}{2(d+5)L_F}\cdot
\frac{\epsilon}{10\overline{R}^2}\cdot 2L_\phi\overline{R}\sqrt{d}\cdot 2\sqrt{d}\overline{R}
\leq\frac{\epsilon}{5},
\]
and
\[
\frac{\eta L_\phi^2(d+5)^3}{K}\sum_{k=0}^{K-1}r(k)^2
\leq 
\left(\frac{1}{2(d+5)L_F}\right)^2
\cdot \frac{\epsilon}{10\overline{R}^2}
\cdot L_\phi^2(d+5)^3\cdot
\frac{4\overline{R}^2}{d+5}
\leq\frac{\epsilon}{10}.
\]
Finally, to bound $\frac{1}{K}\sum_{k=0}^{K-1}\omega_k$, we note that $\Lambda_\phi\leq\Lambda_F$, $\eta L_F\leq 1/(2d)$ and the condition on $r(k)$ imply
\begin{align*}
\omega_k
\leq\ &
2\Lambda_F\left(\overline{R}+\frac{\Lambda_\phi}{2L_F d}\right)
\cdot 3\left(\frac{2\overline{R}}{r(k)}\right)^{\!4}
\exp{\left(\frac{d}{2}-\frac{\delta^2\underline{R}^2}{4r(k)^2}\right)} \\
\leq\ &
6\Lambda_F\left(\overline{R}+\frac{\Lambda_\phi}{2L_F d}\right)
\cdot\left(\frac{4\overline{R}}{\delta\underline{R}}
\sqrt{\frac{d}{2}+4\ln\frac{8\overline{R}}{\underline{R}}+2\ln\frac{d}{\delta^3}}\right)^{\!4}
\cdot \left(\frac{\underline{R}}{8\overline{R}}\right)^{\!4}
\left(\frac{\delta^3}{d}\right)^{\!2} \\
=\ &
\frac{3}{8}\Lambda_F\left(\overline{R}+\frac{\Lambda_\phi}{2L_F d}\right)
\frac{\delta^2}{d^2}
\left(\frac{d}{2}+4\ln\frac{8\overline{R}}{\underline{R}}+2\ln\frac{d}{\delta^3}\right)^{\!2} \\
\leq\ &
\frac{3\epsilon}{40}
\cdot \delta
\left(\frac{1}{2}+\frac{4}{d}\ln\frac{8\overline{R}}{\underline{R}}
+\frac{2}{d}\ln\frac{d}{\delta^3}\right)^{\!2}
\leq\frac{\epsilon}{10}
\end{align*}
as long as $\epsilon$ (and consequently $\delta$) is sufficiently small. We can now put together all previous bounds and conclude that
\[
\min_{1\leq k\leq K}
\mathbb{E}[F(x(k))-F(x^\ast)]
\leq
\mathbb{E}\!\left[
\frac{1}{K}\sum_{k=1}^K
F(x(k))
\right]-F(x^\ast)
\leq\epsilon.
\]

\section{Proof of Theorem~\ref{theorem:2ZFGD_convergence_nonconvex}}
\label{proof_of_theorem2}
We first define some auxiliary quantities and bounds. Define
\begin{equation}\label{eq:grad_mapping_delta_def}
\mathfrak{g}_\delta(x;M) \coloneqq M\left(x-
\mathcal{P}_{(1-\delta)\mathcal{X}}\!\left[x-\frac{1}{M}\nabla F(x)\right]\right).
\end{equation}
By Lemma~\ref{lemma:shrink_proj_difference}, we can relate $\|\mathfrak{g}_\delta(x;M)\|$ with $\|\mathfrak{g}(x;M)\|$ by
\begin{equation}\label{eq:g_delta_bound}
\begin{aligned}
\left\|\mathfrak{g}(x;M)\right\|^2
\leq\ & \frac{5}{4} \left\|\mathfrak{g}_\delta(x;M)\right\|^2
+ 5 \left\|\mathfrak{g}(x;M) - \mathfrak{g}_\delta(x;M)\right\|^2 \\
=\ &
\frac{5}{4} \left\|\mathfrak{g}_\delta(x;M)\right\|^2 + 5
M^2\left\|
\mathcal{P}_{(1-\delta)\mathcal{X}}\!\left[x-\frac{1}{M}\nabla F(x)\right]
-\mathcal{P}_{\mathcal{X}}\!\left[x-\frac{1}{M}\nabla F(x)\right]
\right\|^2 \\
\leq\ &
\frac{5}{4} \left\|\mathfrak{g}_\delta(x;M)\right\|^2
+ 5\delta^2 M^2 \overline{R}^2.
\end{aligned}
\end{equation}
Then, we notice that the $L_F$-smoothnes of $F$ implies that for any fixed $z\in\mathbb{R}^d$, $x\mapsto F(x)+L_F\|x-z\|^2$ is $L_F$-strongly convex. Thus we can define
\[
\hat{x}(k) = \argmin_{y\in(1-\delta)\mathcal{X}}
\!\left(F(y)+L_F\|y-x(k)\|^2\right),
\quad
\hat{F}(x) = \min_{y\in(1-\delta)\mathcal{X}}
\!\left(F(y)+L_F\|y-x\|^2\right).
\]
We next provide a bound on $\left\|
\mathbb{E}[g(k)|\mathcal{F}_k]-\nabla F(x(k))\right\|$, which will be used for subsequent analysis:
\begin{align*}
& \left\|
\mathbb{E}[g(k)|\mathcal{F}_k]-\nabla F(x(k))\right\| \\
=\ &
\left\|
\mathbb{E}[G_\phi(x(k);r(k),z(k))|\mathcal{F}_k]-\nabla \phi(x(k))\right\| \\
\leq\ &
\left\|
\mathbb{E}[G_\phi(x(k);r(k),z(k))|\mathcal{F}_k]-\nabla \phi^{r(k)}(x(k))\right\|
+\left\|\nabla \phi^{r(k)}(x(k))-\nabla \phi(x(k))\right\| \\
\leq\ &
\Lambda_\phi\left(
\frac{4\overline{R}^2}{r(k)}+2
\right)
\exp\!\left(\frac{d}{2}-\frac{\delta^2\underline{R}^2}{4r(k)^2}\right)
+r(k) L_\phi\sqrt{d}
\end{align*}

Our analysis consists of the following steps:

\noindent\textbf{1. Derive a descent property for the iterates}. We introduce the following lemma as our starting point:
\begin{lemma}[{\cite[Eq. (3.11)]{davis2019stochastic}}]
\label{lemma:PSGD_starting}
We have
\[
\|x(k+1)-\hat{x}(k)\|^2
\leq \left\|
(1-2\eta L_F)(x(k)-\hat{x}(k)) - \eta(g(k)-\nabla F(\hat{x}(k)))
\right\|^2.
\]
\end{lemma}
We continue our analysis from the inequality given by Lemma~\ref{lemma:PSGD_starting}. Note that
\begin{align*}
& \mathbb{E}\!\left[\left\|
(1-2\eta L_F)(x(k)-\hat{x}(k)) - \eta(g(k)-\nabla F(\hat{x}(k)))
\right\|^2 |\mathcal{F}_k\right] \\
=\ &
\mathbb{E}\!\left[\left\|
(1-2\eta L_F)(x(k)-\hat{x}(k))
-\eta(\nabla F(x(k))-\nabla F(\hat{x}(k)))
-\eta(g(k)-\nabla F(x(k)))
\right\|^2|\mathcal{F}_k\right] \\
=\ &
\|(1-2\eta L_F)(x(k)-\hat{x}(k))
-\eta(\nabla F(x(k))-\nabla F(\hat{x}(k)))\|^2
+ \eta^2\,\mathbb{E}\!\left[\|g(k)-\nabla F(x(k))\|^2 |\mathcal{F}_k\right] \\
& -
2\eta\langle (1-2\eta L_F)(x(k)-\hat{x}(k))
-\eta(\nabla F(x(k))-\nabla F(\hat{x}(k))),
\mathbb{E}[g(k)|\mathcal{F}_k]-\nabla F(x(k))\rangle.
\end{align*}
Now for the first term, we have
\begin{align*}
& \|(1-2\eta L_F)(x(k)-\hat{x}(k))
-\eta(\nabla F(x(k))-\nabla F(\hat{x}(k)))\|^2 \\
=\ &
(1-2\eta L_F)^2 \|x(k)-\hat{x}(k)\|^2
+\eta^2 \|\nabla F(x(k))-\nabla F(\hat{x}(k))\|^2 \\
&
-2\eta(1-2\eta L_F)\langle
\nabla F(x(k))-\nabla F(\hat{x}(k)),x(k)-\hat{x}(k)\rangle \\
\leq\ &
(1-2\eta L_F)^2 \|x(k)-\hat{x}(k)\|^2
+\eta^2 L_F^2 \|x(k)-\hat{x}(k)\|^2
+ 2\eta(1-2\eta L_F)L_F \|x(k)-\hat{x}(k)\|^2 \\
=\ &
\left(1-\eta L_F\right)^2\|x(k)-\hat{x}(k)\|^2,
\end{align*}
which will further imply
\begin{align*}
& -2\eta\langle (1-2\eta L_F)(x(k)-\hat{x}(k))
-\eta(\nabla F(x(k))-\nabla F(\hat{x}(k))),
\mathbb{E}[g(k)|\mathcal{F}_k]-\nabla F(x(k))\rangle \\
\leq\ &
2\eta\left\|(1-2\eta L_F)(x(k)-\hat{x}(k))
-\eta(\nabla F(x(k))-\nabla F(\hat{x}(k)))\right\|
\cdot
\left\|
\mathbb{E}[g(k)|\mathcal{F}_k]-\nabla F(x(k))\right\| \\
\leq\ &
2\eta(1-\eta L_F)\|x(k)-\hat{x}(k)\|\cdot
\left\|
\mathbb{E}[g(k)|\mathcal{F}_k]-\nabla F(x(k))\right\| \\
\leq\ &
2\eta \|x(k)-\hat{x}(k)\|\cdot \Lambda_\phi\left(
\frac{4\overline{R}^2}{r(k)}+2
\right)
\exp\!\left(\frac{d}{2}-\frac{\delta^2\underline{R}^2}{4r(k)^2}\right)
+2\eta\|x(k)-\hat{x}(k)\|\cdot r(k)L_\phi\sqrt{d} \\
\leq\ &
4\eta\overline{R}\Lambda_\phi\left(
\frac{4\overline{R}^2}{r(k)}+2
\right)
\exp\!\left(\frac{d}{2}-\frac{\delta^2\underline{R}^2}{4r(k)^2}\right)
+\eta L_\phi \|x(k)-\hat{x}(k)\|^2
+\eta L_\phi r(k)^2d.
\end{align*}
We can now apply Lemma~\ref{lemma:PSGD_starting} and combine all the previous bounds with~\eqref{eq:var_gk_bound} to obtain
\begin{align*}
& \mathbb{E}\!\left[\|x(k+1)-\hat{x}(k)\|^2|\mathcal{F}_k\right] \\
\leq\ &
(1-2\eta L_F + \eta L_\phi + \eta^2 L_F^2)
\|x(k)-\hat{x}(k)\|^2
+ 2\eta^2 \Lambda_\phi^2(d+2) \\
&
+ \eta^2 L_\phi^2 r(k)^2(d+5)^3
+
2\eta^2\Lambda_\phi^2
\left[\left(\frac{2\overline{R}}{r(k)}\right)^{\!4}
+\left(\frac{2\overline{R}}{r(k)}\right)^{\!2}
+2
\right]
\exp\left(\frac{d}{2}-\frac{\delta^2\underline{R}^2}{4r(k)^2}\right) \\
&
+4\eta\overline{R}\Lambda_\phi\left(
\frac{4\overline{R}^2}{r(k)}+2
\right)
\exp\!\left(\frac{d}{2}-\frac{\delta^2\underline{R}^2}{4r(k)^2}\right)
+\eta L_\phi r(k)^2d \\
\leq\ &
(1-\eta L_F+\eta^2 L_F^2) \|x(k)-\hat{x}(k)\|^2
+ 2\eta^2\Lambda_\phi^2(d+2)
+\eta L_\phi r(k)^2(\eta L_\phi(d+5)^3+d)+
6\eta\varpi_k,
\end{align*}
where we denote
\[
\varpi_k = \Lambda_\phi(2\overline{R}+\eta\Lambda_\phi)\left(
\frac{2\overline{R}}{r(k)}
\right)^{\!4}\exp\!\left(\frac{d}{2}-\frac{\delta^2\underline{R}^2}{4r(k)^2}\right).
\]
Consequently,
\begin{align*}
\mathbb{E}[\hat{F}(x(k+1))]
\leq\ &
\mathbb{E}\!\left[F(\hat{x}(k)) + L_F\|\hat{x}(k)-x(k+1)\|^2\right] \\
\leq\ &
\mathbb{E}\!\left[F(\hat{x}(k))+ L_F\|x(k)-\hat{x}(k)\|^2\right]
+
\mathbb{E}\!\left[
-\eta L_F^2(1-\eta L_F)\|x(k)-\hat{x}(k)\|^2
\right] \\
&
+2\eta^2L_F\Lambda_\phi^2(d+2)
+\eta L_F L_\phi r(k)^2(\eta L_\phi(d+5)^3+d)
+
6\eta L_F \varpi_k \\
=\ &
\mathbb{E}\!\left[\hat{F}(x(k))\right]
-\eta L_F^2(1-\eta L_F)\mathbb{E}\!\left[\|x(k)-\hat{x}(k)\|^2\right] \\
&+2\eta^2L_F\Lambda_\phi^2(d+2)
+\eta L_F L_\phi r(k)^2(\eta L_\phi(d+5)^3+d)
+
6\eta L_F \varpi_k.
\end{align*}
\noindent\textbf{2. Take the telescoping sum}.
By taking the telescoping sum and using $\eta L_F\leq 1/6$, we can show that
\begin{align*}
\frac{1}{K}\sum_{k=0}^{K-1}
\mathbb{E}\!\left[L_F^2\|x(k)-\hat{x}(k)\|^2\right]
\leq\ &
\frac{
6\left(\hat{F}(x(0))-\mathbb{E}\big[\hat{F}(x(K))\big]\right)
}{5\eta K}
+\frac{12\eta L_F \Lambda_\phi^2(d+2)}{5} \\
&
+\frac{6L_FL_\phi(\eta L_\phi(d+5)^3+d)}{5K}\sum_{k=0}^{K-1}r(k)^2
+\frac{36L_F}{5K}\sum_{k=0}^{K-1}\varpi_k.
\end{align*}
Notice that
\begin{align*}
L_F^2\|x(k)-\hat{x}(k)\|^2
\geq\ &
\frac{1}{2}\|\mathfrak{g}_{\delta}(x(k);2L_F)\|^2
\geq \frac{1}{2}\|\mathfrak{g}_{\delta}(x(k);L_F)\|^2 \\
\geq\ &
\frac{2}{5}\|\mathfrak{g}(x(k);L_F)\|^2
-2\delta^2 L_F^2\overline{R}^2,
\end{align*}
where the first inequality follows from~\cite[Theorem 3.5]{drusvyatskiy2018error},  the second inequality follows from~\cite[Lemma 2]{nesterov2013gradient}, and the last step follows from~\eqref{eq:g_delta_bound}. Moreover, it's not hard to see that $\hat{F}(x(0))\leq F(x(0))$ and
\[
\hat{F}(x(K)) = \min_{y\in(1-\delta)\mathcal{X}}
\left(
F(y)+L_F\|y-x(K)\|^2
\right)
\geq \min_{y\in(1-\delta)\mathcal{X}}
F(y)
\geq \min_{y\in\mathcal{X}}
F(y)=F^\ast.
\]
As a result, we can obtain
\begin{align*}
\frac{1}{K}\sum_{k=0}^{K-1}
\mathbb{E}\!\left[
\|\mathfrak{g}(x(k);L_F)\|^2
\right]
\leq\ &
\frac{3(F(x(0))-F^\ast)}{\eta K}
+6\eta L_F\Lambda_\phi^2(d+2)
+5\delta^2 L_F^2\overline{R}^2
\\
&
+
\frac{3L_FL_\phi(\eta L_\phi(d+5)^3+d)}{K}
\sum_{k=0}^{K-1}r(k)^2
+\frac{18 L_F}{K}\sum_{k=0}^{K-1}\varpi_k.
\end{align*}
We can now apply the conditions on the algorithmic parameters and obtain
\[
5\delta^2 L_F^2\overline{R}^2
\leq\frac{\epsilon}{5},
\qquad
6\eta L_F\Lambda_\phi^2(d+2)
\leq\frac{\epsilon}{5},
\qquad
\frac{3(F(x(0))-F^\ast)}{\eta\epsilon}
\leq\frac{\epsilon}{5},
\]
and
\begin{align*}
\frac{3L_FL_\phi(\eta L_\phi(d+5)^3+d)}{K}\sum_{k=0}^{K-1}r(k)^2
\leq\ &
\frac{3L_FL_\phi((d+5)^2+d)}{15(F(x(0))-F^\ast)}\cdot\frac{\epsilon}{L_F(d+5)}
\cdot\frac{F(x(0))-F^\ast}{L_\phi(d+6)}
\leq\frac{\epsilon}{5},
\end{align*}
and for sufficiently small $\epsilon$,
\begin{align*}
18L_F\varpi_k
\leq\ &
36L_F\Lambda_\phi
\left(\overline{R}+\frac{\Lambda_\phi}{2L_Fd}\right)
\left(
\frac{4\overline{R}}{\delta\underline{R}}
\sqrt{\frac{d}{2}+4\ln\frac{8\overline{R}}{\underline{R}}
+\ln\frac{d^3}{\delta^7}}
\right)^{\!4}
\cdot\left(\frac{\underline{R}}{8\overline{R}}\right)^{\!4}\cdot\frac{\delta^7}{d^3} \\
=\ &
\frac{9}{2} \cdot\frac{L_F\Lambda_\phi}{2d}\left(
\overline{R}+\frac{\Lambda_\phi}{2L_F d}
\right)
\frac{\delta^3}{d^2}
\left(\frac{d}{2}+4\ln\frac{8\overline{R}}{\underline{R}}
+\ln\frac{d^3}{\delta^7}\right)^{\!2} \\
\leq\ &
\frac{9\epsilon}{50}
\cdot\delta\left(\frac{1}{2}
+\frac{4}{d}\ln\frac{8\overline{R}}{\underline{R}}+\frac{1}{d}\ln\frac{d^3}{\delta^7}\right)
\leq\frac{\epsilon}{5}.
\end{align*}
The proof is now complete.

\section{Proof of the Equality~\eqref{eq:variance_two_point_zero_r}}
\label{appendix:proof_variance_two_point_zero_r}

Note that by~\eqref{eq:variance_bound_temp1}, we have
\[
\lim_{r\downarrow 0}
\mathbb{E}_{z\sim\mathcal{Z}(x,r)}
\!\left[
\left\|
G_h(x;r,z)-\langle\nabla h(x),z\rangle z
\right\|^2
\right]
=0.
\]
Therefore we only need to show that
\[
\lim_{r\downarrow 0}
\mathbb{E}_{z\sim\mathcal{Z}(x,r)}\!\left[
\|\langle\nabla h(x),z\rangle z
-\nabla h(x)
\|^2
\right]
=(d+1)\|\nabla h(x)\|^2.
\]
Indeed, we have
\begin{align*}
& \mathbb{E}_{z\sim\mathcal{Z}(x,r)}\!\left[
\|\langle\nabla h(x),z\rangle z
-\nabla h(x)
\|^2
\right] - (d+1)\|\nabla h(x)\|^2 \\
=\ &
\nabla h(x)^\tran\mathbb{E}_{z\sim\mathcal{Z}(x,r)}\!\left[
(zz^\tran-I)^2
-(d+1)I\right]\nabla h(x) \\
=\ &
\nabla h(x)^\tran\mathbb{E}_{z\sim\mathcal{Z}(x,r)}\!\left[
(\|z\|^2 zz^\tran-(d+2)I) - 2(zz^\tran - I)
\right]\nabla h(x) \\
=\ &
\nabla h(x)^\tran
\left(\mathbb{E}_{z\sim\mathcal{Z}(x,r)}
\!\left[\|z\|^2 zz^\tran\right]
-\mathbb{E}_{z\sim\mathcal{N}(0,I_d)}
\!\left[\|z\|^2 zz^\tran\right]\right)\nabla h(x) \\
&
-2\nabla h(x)^\tran \left(\mathbb{E}_{z\sim\mathcal{Z}(x,r)}
\!\left[zz^\tran\right]
-\mathbb{E}_{z\sim\mathcal{N}(0,I_d)}
\!\left[zz^\tran\right]\right)\nabla h(x) \\
=\ &
\nabla h(x)^\tran
\mathbb{E}_{z\sim\mathcal{N}(0,I_d)}
\!\left[\|\mathcal{P}z\|^2 \cdot \mathcal{P}z\cdot \mathcal{P}z^\tran
-\|z\|^2 zz^\tran\right]\nabla h(x) \\
&
-2\nabla h(x)^\tran \mathbb{E}_{z\sim\mathcal{N}(0,I_d)}
\!\left[\mathcal{P}z\cdot \mathcal{P}z^\tran
-zz^\tran\right]\nabla h(x),
\end{align*}
where we used $\mathbb{E}_{z\sim\mathcal{N}(0,I_d)}[\|z\|^2 zz^\tran]=(d+2)I_d$, and $\mathcal{P}$ denotes the projection onto the convex set $S(x,r)=\prod_{i=1}^N S_i(x_i,r)$. Note that as $r\downarrow 0$, we have $\mathcal{P}z\rightarrow z$ and $\|\mathcal{P}z\|\leq\|z\|$ for every $z\in\mathbb{R}^d$. By using the dominated convergence theorem, we see that
\[
\lim_{r\downarrow 0}
\mathbb{E}_{z\sim\mathcal{N}(0,I_d)}
\!\left[\|\mathcal{P}z\|^2 \cdot \mathcal{P}z\cdot \mathcal{P}z^\tran
-\|z\|^2 zz^\tran\right]
=0,
\]
and
\[
\lim_{r\downarrow 0}
\mathbb{E}_{z\sim\mathcal{N}(0,I_d)}
\!\left[\mathcal{P}z\cdot \mathcal{P}z^\tran
-zz^\tran\right]
=0.
\]
As a result,
\[
\lim_{r\downarrow 0}
\left(\mathbb{E}_{z\sim\mathcal{Z}(x,r)}\!\left[
\|\langle\nabla h(x),z\rangle z
-\nabla h(x)
\|^2
\right] - (d+1)\|\nabla h(x)\|^2\right)=0,
\]
which completes the proof.